%% file: Coordination.tex
\documentclass[]{IEEEtran}
\usepackage{amsfonts}
\usepackage{amssymb}
\usepackage{amsmath}
\usepackage{amsthm}
\usepackage{psfrag}
\usepackage{graphicx}
\usepackage{subfigure}

\newcommand{\argmin}{\operatornamewithlimits{argmin}}
\newcommand{\TS}[1]{\ensuremath{{\cal T}_{\epsilon_{#1}}^{(n)}}}

\newtheorem{theorem}{Theorem}
\newtheorem{lemma}[theorem]{Lemma}

\newtheorem{conjecture}{Conjecture}
\newtheorem{example}{Example}
\newtheorem{definition}{Definition}

\begin{document}

\title{Coordination Capacity}

\author{
Paul Cuff, {\em Member, IEEE}, Haim Permuter, {\em Member, IEEE}, and Thomas M. Cover, {\em Fellow, IEEE}
\thanks{This work was partially supported by the National Science Foundation (NSF) through the grant CCF-0635318.}
\thanks{Paul Cuff is with the Department of Electrical Engineering at Princeton University and can be reached at cuff@princeton.edu.}
\thanks{Haim Permuter is with the Department of Electrical Engineering at Ben Gurion University and can be reached at haimp@bgu.ac.il.}
\thanks{Thomas M. Cover is with the Department of Electrical Engineering at Stanford University and can be reached at cover@stanford.edu.}
}

\maketitle

\begin{abstract}
We develop elements of a theory of cooperation and coordination in networks.  Rather than considering a communication network as a means of distributing information, or of reconstructing random processes at remote nodes, we ask what dependence can be established among the nodes given the communication constraints.  Specifically, in a network with communication rates $\{R_{i,j}\}$ between the nodes, we ask what is the set of all achievable joint distributions $p(x_1,...,x_m)$ of actions at the nodes of the network.  Several networks are solved, including arbitrarily large cascade networks.

Distributed cooperation can be the solution to many problems such as distributed games, distributed control, and establishing mutual information bounds on the influence of one part of a physical system on another.
\end{abstract}

\begin{keywords}
Common randomness, cooperation capacity, coordination capacity, network dependence, rate distortion, source coding, strong Markov lemma, task assignment, Wyner common information.
\end{keywords}

\input{coordination_intro.tex}
\input{coordination_problem_specifics.tex}
\input{coordination_complete_results.tex}
\input{coordination_partial_results.tex}
\input{coordination_strong_coordination.tex}
\input{coordination_rate_distortion.tex}
\input{coordination_appendix.tex}
\input{coordination_conclusion.tex}

\bibliographystyle{unsrt}

\newpage

\begin{biographynophoto}{Paul Cuff}
(S'08-M'10) received the B.S. degree in electrical engineering from Brigham Young University in 2004 and the M.S and Ph.D. degrees in electrical engineering from Stanford University in 2006 and 2009.  He was awarded the ISIT 2008 Student Paper Award for his work titled "Communication Requirements for Generating Correlated Random Variables" and was a recipient of the National Defense Science and Engineering Graduate Fellowship and the Numerical Technologies Fellowship.

Dr. Cuff is an Assistant Professor of Electrical Engineering at Princeton University.
\end{biographynophoto}

\begin{biographynophoto}{Haim Permuter}
(M'08) received his B.Sc. (summa cum laude) and M.Sc. (summa cum laude) degree in Electrical and Computer Engineering from the Ben-Gurion University, Israel, in 1997 and 2003, respectively, and Ph.D. degrees in Electrical  Engineering from Stanford University, California in 2008. Between 1997 and 2004, he was an officer at a research and development unit of the Israeli Defense Forces. He is currently a Lecturer at Ben-Gurion university. He is a recipient of the Fullbright Fellowship, the Stanford Graduate Fellowship (SGF), Allon Fellowship, and Bergman award.
\end{biographynophoto}

\begin{biographynophoto}{Thomas M. Cover,}
the K.T. Li Professor of Electrical Engineering and Professor of Statistics at Stanford, does research in information theory, communication theory and statistics, and is the coauthor of the textbook, Elements of Information Theory. He was Lab Director of the Information Systems Laboratory in Electrical Engineering from 1989 to 1996. He has been the contract statistician for the California State Lottery and a consultant to AT\&T Laboratories and IBM. He received the 1990 Claude E. Shannon Award in information theory and has also received the IEEE Neural Network Council's Pioneer Award in 1993 for his work on the capacity of neural nets. He received the 1997 IEEE Richard M. Hamming medal for contributions to information, communication theory and statistics and is a member of the National Academy of Engineering and the American Academy of Arts and Sciences. He is currently working on network information theory and the interplay between information theory and investment.
\end{biographynophoto}

\vfill

\end{document}

%% file: coordination_intro.tex
\section{Introduction}
\label{section introduction}

\PARstart{C}{ommunication} is required to establish cooperative behavior.  In a network of nodes where relevant information is known at only some nodes in the network, finding the minimum communication requirements to coordinate actions can be posed as a network source coding problem.  This diverges from traditional source coding.  Rather than focus on sending data from one point to another with a fidelity constraint, we consider the communication needed to establish coordination summarized by a joint probability distribution of behavior among all nodes in the network.

A large variety of research addresses the challenge of collecting or moving information in networks.  Network coding \cite{network_coding} seeks to efficiently move independent flows of information over shared communication links.  On the other hand, distributed average consensus \cite{tsitsiklis86} involves collecting related information.  Sensors in a network collectively compute the average of their measurements in a distributed fashion.  The network topology and dynamics determine how many rounds of communication among neighbors are needed to converge to the average and how good the estimate will be at each node \cite{xiao-boyd-kim}.  Similarly, in the gossiping Dons problem \cite{bollobas}, each node starts with a unique piece of gossip, and one wishes to know how many exchanges of gossip are required to make everything known to everyone.  Computing functions in a network is considered in \cite{yao79}, \cite{orlitsky90}, and \cite{ayaso-shah-dahleh08}.

Our work, introduced in \cite{Cover07Permuter}, has several distinctions from the network communication examples mentioned.  First, we keep the purpose for communication very general, which means sometimes we get away with saying very little about the information in the network while still achieving the desired coordination.  We are concerned with the joint distribution of actions taken at the various nodes in the network, and the ``information'' that enters the network is nothing more than actions  that are selected randomly by nature and assigned to certain nodes. Secondly, we consider quantization and rates of communication in the network, as opposed to only counting the number of exchanges.  We find that we can gain efficiency by using vector quantization specifically tailored to the network topology.

Figure \ref{figure general network} shows an example of a network with rate-limited communication links.  In general, each node in the network performs an action where some of these actions are selected randomly by nature.  In this example, the source set ${\cal S}$ indicates which actions are chosen by nature:  Actions $X_1$, $X_2$, and $X_3$ are assigned randomly according to the joint distribution $p_0(x_1,x_2,x_3)$.  Then, using the communication and common randomness that is available to all nodes, the actions $Y_1$, $Y_2$, and $Y_3$ outside of ${\cal S}$ are produced.  We ask, which conditional distributions $p(y_1,y_2,y_3|x_1,x_2,x_3)$ are compatible with the network constraints.

\begin{figure}[h]
\psfrag{x1}[][][1]{$X_1$}
\psfrag{x2}[][][1]{$X_2$}
\psfrag{x3}[][][1]{$X_3$}
\psfrag{y1}[][][1]{$Y_1$}
\psfrag{y2}[][][1]{$Y_2$}
\psfrag{y3}[][][1]{$Y_3$}
\psfrag{r1}[][][.8]{$R_1$}
\psfrag{r2}[][][.8]{$R_2$}
\psfrag{r3}[][][.8]{$R_3$}
\psfrag{r4}[][][.8]{$R_4$}
\psfrag{r5}[][][.8]{$R_5$}
\psfrag{r6}[][][.8]{$R_6$}
\psfrag{r7}[][][.8]{$R_7$}
\psfrag{p}[][][.8]{$p_0(x_1,x_2,x_3)$}
\psfrag{s}[][][1]{${\cal S}$}
\psfrag{region}[][][1]{${\cal P}_{p_0}({\cal R})=\{p(y_1,y_2,y_3|x_1,x_2,x_3)\}$}
\centerline{\includegraphics[width=8cm]{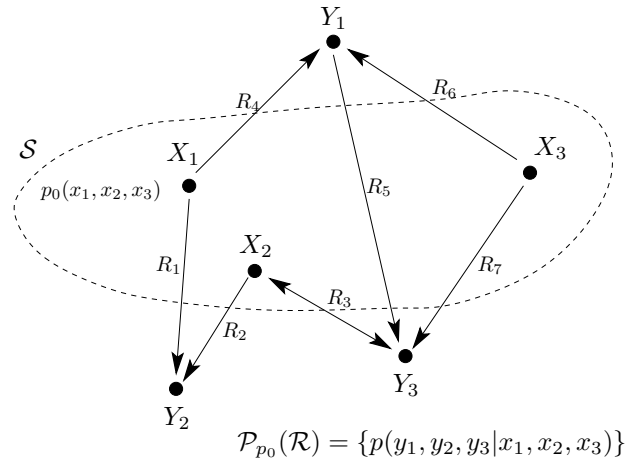}}
\caption{{\em Coordination capacity.}  This network represents the general framework we consider.  The nodes in this network have rate-limited links of communication between them.  Each node performs an action.  The actions $X_1$, $X_2$, and $X_3$ in the source set ${\cal S}$ are chosen randomly by nature according to $p_0(x_1,x_2,x_3)$, while the actions $Y_1$, $Y_2$, and $Y_3$ are produced based on the communication and common randomness in the network.  What joint distributions $p_0(x_1,x_2,x_3)p(y_1,y_2,y_3|x_1,x_2,x_3)$ can be achieved?}
\label{figure general network}
\end{figure}

A variety of applications are encompassed in this framework.  This could be used to model sensors in a sensor network, sharing information in the standard sense, while also cooperating in their transmission of data.  Similarly, a wireless ad hoc network can improve performance by cooperating among nodes to allow beam-forming and interference alignment.  On the other hand, some settings do not involve moving information in the usual sense.  The nodes in the network might comprise a distributed control system, where the behavior at each node must be related to the behavior at other nodes and the information coming into the system.  Also, with computing technology continuing to move in the direction of parallel processing, even across large networks, a network of computers must coherently perform computations while distributing the work load across the participating machines.  Alternatively, the nodes might each be agents taking actions in a multiplayer game.

Network communication can be revisited from the viewpoint of coordinated actions.  Rate distortion theory becomes a special case.  More generally, we ask how we can build dependence among the nodes. What is it good for? How do we use it?

In this paper we deal with two fundamentally different notions of coordination which we distinguish as {\em empirical coordination} and {\em strong coordination}, both associated with a desired joint distribution of actions.  Empirical coordination is achieved if the joint type of the actions in the network---the empirical joint distribution---is close to the desired distribution.  Techniques from rate-distortion theory are relevant here.  Strong coordination instead deals with the joint probability distribution of the actions.  If the actions in the network are generated randomly so that a statistician cannot reliably distinguish (as measured by total variation) between the constructed $n$-length sequence of actions and random samples from the desired distribution, then strong coordination is achieved.  The approach and proofs in this framework are related to the common information work by Wyner \cite{wyner}.

Before developing the mathematical formulation, consider the first surprising observation.

{\it No communication:} Suppose we have three nodes choosing actions and no communication is allowed between the nodes (Fig. \ref{f_all_dist}).  We assume that common randomness is available to all the nodes.  What is the set of joint distributions $p(x,y,z)$ that can be achieved at these isolated nodes?  The answer turns out to be any joint distribution whatsoever.  The nodes can agree ahead of time on how they will behave in the presence of common randomness (for example, a time stamp used as a seed for a random number generator). Any triple of random variables can be created as functions of common randomness.

This would seem to be the end of the problem, but the problem changes dramatically when one of the nodes is specified by nature to take on a certain value, as will be the case in each of the scenarios following.

\begin{figure}[h]
\psfrag{X}[][][1]{$\;\;\;\;\;\;\;\;\;\;\;\;\;\;\; X=X(\omega)$}
\psfrag{Y}[][][1]{$\;\;\;\;\;\;\;\;\;\;\;\;\;\;\; Y=Y(\omega)$}
\psfrag{Z}[][][1]{$\;\;\;\;\;\;\;\;\;\;\;\;\;\;\; Z=Z(\omega)$}
\psfrag{Dist}[][][1]{$\;\;\;\;\;\;\;\;\;\;\;\;\;\;\; {\cal P} = \{p(x,y,z)\}$}
\centerline{\includegraphics[width=6cm]{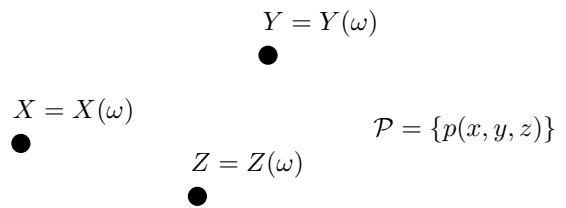}}
\caption{{\em No communication.}  Any distribution $p(x,y,z)$ can be achieved without communication between nodes.  Define three random variables $X(\cdot)$, $Y(\cdot)$, and $Z(\cdot)$ with the appropriate joint distribution, on the standard probability space $(\Omega, {\cal B}, {\cal P})$, and let the actions at the nodes be $X(\omega)$, $Y(\omega)$, and $Z(\omega)$, where $\omega \in \Omega$ is the common randomness.}
\label{f_all_dist}
\end{figure}

An eclectic collection of work, ranging from game theory to quantum information theory, has a number of close relationships to our approach and results.  For example, Anantharam and Borkar \cite{AnantharamBorkar2007} let two agents generate actions for a multiplayer game based on correlated observations and common randomness and ask what kind of correlated actions are achievable.  From a quantum mechanics perspective, Barnum et. al. \cite{Barnum_quantum_mixed_states} consider quantum coding of mixed quantum states.  Kramer and Savari \cite{kramer_savari07} look at communication for the purpose of ``communicating probability distributions'' in the sense that they care about reconstructing a sequence with the proper empirical distribution of the sources rather than the sources themselves.  Weissman and Ordentlich \cite{weissman-ordentlich} make statements about the empirical distributions of sub-blocks of source and reconstruction symbols in a rate-constrained setting.  And Han and Verd\'{u} \cite{han} consider generating a random process via use of a memoryless channel, while Bennett et. al. \cite{bennett2002} propose a ``reverse Shannon theorem'' stating the amount of noise free communication necessary to synthesize a memoryless channel.

In this work, we consider coordination of actions in two and three node networks.  These serve as building blocks for understanding larger networks.  Some of the actions at the nodes are given by nature, and some are constructed by the node itself.  We describe the problem precisely in Section \ref{section weak coordination}.  For some network settings we characterize the entire solution, but for others we give partial results including bounds and solutions to special cases.  The complete results are presented in Section \ref{section complete results} and include a variant of the multiterminal source coding problem.  Among the partial results of Section \ref{section partial results}, a consistent trend in coordination strategies is identified, and the golden ratio makes a surprise appearance.

In Section \ref{section strong coordination} we consider strong coordination.  We characterize the communication requirements in a couple of settings and discuss the role of common randomness.  If common randomness is available to all nodes in the network then empirical coordination and strong coordination seem to require equivalent communication resources, consistent with the implications of the ``reverse Shannon theorem'' \cite{bennett2002}.  Furthermore, we can quantify the amount of common randomness needed, treating common randomness itself as a scarce resource.

Rate-distortion regions are shown to be projections of the coordination capacity region in Section \ref{section rate distortion}.  The proofs for all theorems are presented together in Section \ref{section proofs}, where we introduce a stronger Markov Lemma (Theorem \ref{theorem strong markov}) that may be broadly useful in network information theory.  In our closing remarks we show cases where this work can be extrapolated to large networks to identify the efficiency of different network topologies.

%% file: coordination_problem_specifics.tex
\section{Empirical Coordination}
\label{section weak coordination}

In this section and the next we address questions of the following nature:  If three different tasks are to be  performed in a shared effort between three people, but one person is randomly assigned his responsibility, how much must he tell the others about his assignment in order to divide the labor?

\subsection{Problem specifics}
\label{subsection definitions}

The definitions in this section pinpoint the concept of {\em empirical coordination}.  We will consider coordination in a variety of two and three node networks.  The basic meaning of empirical coordination is the same for each network---we use the network communication to construct a sequence of actions that have an empirical joint distribution closely matching a desired distribution.  What's different from one problem to the next is the set of nodes whose actions are selected randomly by nature and the communication limitations imposed by the network topology.

Here we define the problem in the context of the cascade network of Section \ref{subsection cascade} shown in Figure \ref{figure block diagram}.  These definitions have obvious generalizations to other networks.

\begin{figure}[h]
\psfrag{l1}[][][0.7]{Node X}
\psfrag{l2}[][][0.6]{}
\psfrag{l3}[][][0.6]{$i(X^n,\omega)$}
\psfrag{l4}[][][0.7]{Node Y}
\psfrag{l5}[][][0.6]{$y^n(I,\omega)$}
\psfrag{l6}[][][0.6]{$j(I,\omega)$}
\psfrag{l7}[][][0.7]{Node Z}
\psfrag{l8}[][][0.6]{$z^n(J,\omega)$}
\psfrag{l9}[][][0.6]{}
\psfrag{l10}[][][0.6]{$X^n$}
\psfrag{l11}[][][0.6]{$I \in [2^{nR_1}]$}
\psfrag{l12}[][][0.6]{$J \in [2^{nR_2}]$}
\psfrag{l13}[][][0.6]{$Y^n$}
\psfrag{l14}[][][0.6]{$Z^n$}
\psfrag{l15}[][][0.6]{$\sim \prod_{i=1}^n p_0(x_i)$}
\centerline{\includegraphics[width=.45\textwidth]{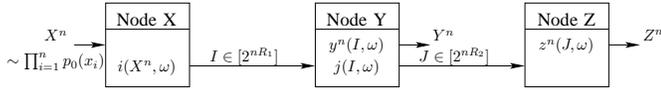}}
\caption{{\em Cascade network.}  Node X is assigned actions $X^n$ chosen by nature according to $p(x^n)=\prod_{i=1}^n p_0(x_i)$.  A message $I$ in the set $\{1,...,2^{nR_1}\}$ is constructed based on $X^n$ and the common randomness $\omega$ and sent to Node Y, which constructs both an action sequence $Y^n$ and a message $J$ in the set $\{1,...,2^{nR_2}\}$.  Finally, Node Z produces actions $Z^n$ based on the message $J$ and the common randomness $\omega$.  This is summarized in Figure \ref{figure shorthand}.}
\label{figure block diagram}
\end{figure}

\begin{figure}[h]
\psfrag{l1}[][][.8]{$X \sim p_0(x)$}
\psfrag{l2}[][][.8]{$Y$}
\psfrag{l3}[][][.8]{$Z$}
\psfrag{l4}[][][.8]{$R_1$}
\psfrag{l5}[][][.8]{$R_2$}
\psfrag{l6}[][][.7]{Node X}
\psfrag{l7}[][][.7]{Node Y}
\psfrag{l8}[][][.7]{Node Z}
\centerline{\includegraphics[width=.4\textwidth]{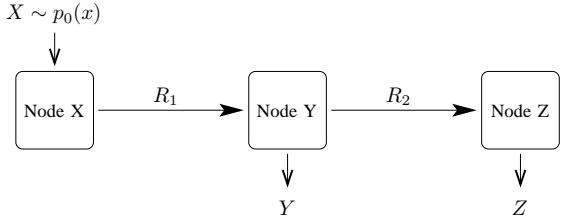}}
\caption{Shorthand notation for the cascade network of Figure \ref{figure block diagram}.}
\label{figure shorthand}
\end{figure}

In the cascade network of Figure \ref{figure block diagram}, node $X$ has a sequence of actions $X_1, X_2, ...$ specified randomly by nature.  Note that a node is allowed to see all of its actions before it summarizes them for the next node.  Communication is used to give Node $Y$ and Node $Z$ enough information to choose sequences of actions that are empirically correlated with $X_1, X_2,...$ according to a desired joint distribution $p_0(x)p(y,z|x)$.  The communication travels in a cascade, first from Node $X$ to Node $Y$ at rate $R_1$ bits per action, and then from Node $Y$ to Node $Z$ at rate $R_2$ bits per action.

Specifically, a $(2^{nR_1},2^{nR_2},n)$ {\em coordination code} is used as a protocol to coordinate the actions in the network for a block of $n$ time periods.  The coordination code and the distribution of the random actions $X^n$ induce a joint distribution on the actions in the network.  If the joint type of the actions in the network can be made arbitrarily close to a desired distribution $p_0(x)p(y,z|x)$ with high probability, as dictated by the distribution induced by a $(2^{nR_1},2^{nR_2},n))$ coordination code, then $p_0(x)p(y,z|x)$ is achievable with the rate pair $(R_1,R_2)$.

\begin{definition}[Coordination code]
\label{definition coordination code}
A $(2^{nR_1},2^{nR_2},n)$ coordination code for the cascade network of Figure \ref{figure block diagram} consists of four functions---an encoding function
\begin{eqnarray*}
i & : & {\cal X}^n \times \Omega \longrightarrow \{1,...,2^{nR_1}\},
\end{eqnarray*}
a recoding function
\begin{eqnarray*}
j & : & \{1,...,2^{nR_1}\} \times \Omega \longrightarrow \{1,...,2^{nR_2}\},
\end{eqnarray*}
and two decoding functions
\begin{eqnarray*}
y^n & : & \{1,...,2^{nR_1}\} \times \Omega \longrightarrow {\cal Y}^n, \\
z^n & : & \{1,...,2^{nR_2}\} \times \Omega \longrightarrow {\cal Z}^n.
\end{eqnarray*}
\end{definition}

\begin{definition}[Induced distribution]
\label{definition induced distribution}
The induced distribution $\tilde{p}(x^n,y^n,z^n)$ is the resulting joint distribution of the actions in the network $X^n$, $Y^n$, and $Z^n$ when a $(2^{nR_1},2^{R_2},n)$ coordination code is used. \end{definition}

Specifically, the actions $X^n$ are chosen by nature i.i.d. according to $p_0(x)$ and independent of the common randomness $\omega$.  Thus, $X^n$ and $\omega$ are jointly distributed according to a product distribution,
\begin{eqnarray*}
(X^n,\omega) & \sim & p(\omega) \prod_{i=1}^n p_0(x_i).
\end{eqnarray*}
The actions $Y^n$ and $Z^n$ are functions of $X^n$ and $\omega$ given by implementing the coordination code as
\begin{eqnarray*}
Y^n & = & y^n(i(X^n,\omega),\omega), \\
Z^n & = & z^n(j(i(X^n,\omega),\omega),\omega).
\end{eqnarray*}

\begin{definition}[Joint type]
\label{definition joint type}
The joint type $P_{x^n,y^n,z^n}$ of a tuple of sequences $(x^n,y^n,z^n)$ is the empirical probability mass function, given by
\begin{eqnarray*}
P_{x^n,y^n,z^n}(x,y,z) & \triangleq & \frac{1}{n} \sum_{i=1}^n {\bf 1} ((x_i,y_i,z_i)=(x,y,z)),
\end{eqnarray*}
for all $(x,y,z) \in {\cal X} \times {\cal Y} \times {\cal Z}$, where ${\bf 1}$ is the indicator function.
\end{definition}

\begin{definition}[Total variation]
\label{definition total variation}
The total variation between two probability mass functions is half the $L_1$ distance between them, given by
\begin{eqnarray*}
\|p(x,y,z) - q(x,y,z)\|_{TV} & \triangleq & \frac{1}{2} \sum_{x,y,z} |p(x,y,z) - q(x,y,z)|.
\end{eqnarray*}
\end{definition}

\begin{definition}[Achievability]
\label{definition achievability}
A desired distribution $p_0(x)p(y,z|x)$ is achievable for empirical coordination with the rate pair $(R_1,R_2)$ if there exists a sequence of $(2^{nR_1},2^{nR_2},n)$ coordination codes and a choice of $p(\omega)$ such that the total variation between the joint type of the actions in the network and the desired distribution goes to zero in probability (under the induced distribution).  That is,
\begin{eqnarray*}
\left\| P_{X^n,Y^n,Z^n}(x,y,z)-p_0(x)p(y,z|x) \right\|_{TV} \longrightarrow 0 \mbox{ in probability}.
\end{eqnarray*}
\end{definition}

We now define the region of all rate-distribution pairs in Definition \ref{definition coordination capacity region} and slice it into rates for a given distribution in Definition \ref{definition coordination rate} and distributions for a given set of rates in Definition \ref{definition rate coordination}.

\begin{definition}[Coordination capacity region]
\label{definition coordination capacity region}
The coordination capacity region ${\cal C}_{p_0}$ for the source distribution $p_0(x)$ is the closure of the set of rate-coordination tuples $(R_1,R_2,p(y,z|x))$ that are achievable:
\begin{eqnarray*}
\label{equation inner bound}
{\cal C}_{p_0} \triangleq {\bf Cl} \left\{
\begin{array}{l}
(R_1,R_2,p(y,z|x)) \; : \\
p_0(x)p(y,z|x) \mbox{ is achievable at rates } (R_1,R_2)
\end{array}
\right\}.
\end{eqnarray*}
\end{definition}

\begin{definition}[Rate-coordination region]
\label{definition coordination rate}
The rate-coordination region ${\cal R}_{p_0}$ is a slice of the coordination capacity region corresponding to a fixed distribution $p(y,z|x)$:
\begin{eqnarray*}
{\cal R}_{p_0}(p(y,z|x)) & \triangleq & \{ (R_1,R_2) \; : \; (R_1,R_2,p(y,z|x)) \in {\cal C}_{p_0} \}.
\end{eqnarray*}
\end{definition}

\begin{definition}[Coordination-rate region]
\label{definition rate coordination}
The coordination-rate region ${\cal P}_{p_0}$ is a slice of the coordination capacity region corresponding to a tuple of rates $(R_1,R_2)$:
\begin{eqnarray*}
{\cal P}_{p_0}(R_1,R_2) & \triangleq & \{ p(y,z|x) \; : \; (R_1,R_2,p(y,z|x)) \in {\cal C}_{p_0} \}.
\end{eqnarray*}
\end{definition}

\subsection{Preliminary observations}
\label{subsection lemmas}

\begin{lemma}[Convexity of coordination]
\label{lemma convexity}
${\cal C}_{p_0}$, ${\cal R}_{p_0}$, and ${\cal P}_{p_0}$ are all convex sets.
\end{lemma}

\begin{proof}
The coordination capacity region ${\cal C}_{p_0}$ is convex because time-sharing can be used to achieve any point on the chord between two achievable rate-coordination pairs.  Simply combine two sequences of coordination codes that achieve the two points in the coordination capacity region by using one code and then the other in a proportionate manner to achieve any point on the chord.  The definition of joint type in Definition \ref{definition joint type} involves an average over time. Thus if one sequence is concatenated with another sequence, the resulting joint type is a weighted average of the joint types of the two composing sequences.  Rates of communication also combine according to the same weighted average.  The rate of the resulting concatenated code is the weighted average of the two rates.

The rate-coordination region ${\cal R}_{p_0}$ is the intersection of the coordination capacity region ${\cal C}_{p_0}$ with a hyperplane, which are both convex sets.  Likewise for the coordination-rate region ${\cal P}_{p_0}$.  Therefore, ${\cal R}_{p_0}$ and ${\cal P}_{p_0}$ are both convex.
\end{proof}

Common randomness used in conjunction with randomized encoders and decoders can be a crucial ingredient for some communication settings, such as secure communication.  We see, for example, in Section \ref{section strong coordination} that common randomness is a valuable resource for achieving strong coordination.  However, it does not play a necessary role in achieving empirical coordination, as the following theorem shows.

\begin{theorem}[Common randomness doesn't help]
\label{theorem no common randomness}
Any desired distribution $p_0(x)p(y,z|x)$ that is achievable for empirical coordination with the rate pair $(R_1,R_2)$ can be achieved with $\Omega = \emptyset$.
\end{theorem}

\begin{proof}
Suppose that $p_0(x)p(y,z|x)$ is achievable for empirical coordination with the rate pair $(R_1,R_2)$.  Then there exists a sequence of $(2^{nR_1},2^{nR_2},n)$ coordination codes for which the expected total variation between the joint type and $p(x,y,z)$ goes to zero with respect to the induced distribution.  This follows from the bounded convergence theorem since total variation is bounded by one.  By iterated expectation,
\begin{eqnarray*}
{\bf E} \left[ {\bf E} \left[ \left\| P_{X^n,Y^n,Z^n} - p_0(x)p(y,z|x) \right\|_{TV} | \omega \right] \right] \;\; =& & \\
{\bf E} \left\| P_{X^n,Y^n,Z^n} - p_0(x)p(y,z|x) \right\|_{TV}. & &
\end{eqnarray*}
Therefore, there exists a value $\omega^*$ such that
\begin{eqnarray*}
{\bf E} \left[ \left\| P_{X^n,Y^n,Z^n} - p_0(x)p(y,z|x) \right\|_{TV} | \omega^* \right] \;\; \leq & & \\
{\bf E} \left\| P_{X^n,Y^n,Z^n} - p_0(x)p(y,z|x) \right\|_{TV}. & &
\end{eqnarray*}

Define a new coordination code that doesn't depend on $\omega$ and at the same time doesn't increase the expected total variation:
\begin{eqnarray*}
i^*(x^n) & = & i(x^n, \omega^*), \\
j^*(i) & = & j(i,\omega^*), \\
y^{n*}(i) & = & Y^n(i, \omega^*), \\
z^{n*}(j) & = & Z^n(j,\omega^*).
\end{eqnarray*}

This can be done for each $(2^{nR_1},2^{nR_2},n)$ coordination code for $n=1,2,...$.
\end{proof}

\subsection{Generalization}
\label{subsection generalize}

We investigate empirical coordination in a variety of networks in Sections \ref{section complete results} and \ref{section partial results}.  In each case, we explicitly specify the structure and implementation of the coordination codes, similar to Definitions \ref{definition coordination code} and \ref{definition induced distribution}, while all other definitions carry over in a straightforward manner.

We use a {\em shorthand} notation in order to illustrate each network setting with a simple and consistent figure.  Figure \ref{figure shorthand} shows the shorthand notation for the cascade network of Figure \ref{figure block diagram}.  The random actions that are specified by nature are shown with arrows pointing down toward the node (represented by a block).  Actions constructed by the nodes themselves are shown coming out of the node with an arrow downward.  And arrows indicating communication from one node to another are labeled with the rate limits for the communication along those links.

%% file: coordination_complete_results.tex
\section{Coordination---Complete results}
\label{section complete results}

In this section we present the coordination capacity regions ${\cal C}_{p_0}$ for empirical coordination in four network settings:  a network of two nodes; a cascade network; an isolated node network; and a degraded source network.  Proofs are left to Section \ref{section proofs}.  As a consequence of Theorem \ref{theorem no common randomness} we need not use common randomness.  Common randomness will only be required when we try to generate desired distributions over entire n-blocks in Section \ref{section strong coordination}.

\subsection{Two nodes}
\label{subsection two nodes}

In the simplest network setting shown in Figure \ref{figure two nodes}, we consider two nodes, X and Y.  The action $X$ is specified by nature according to $p_0(x)$, and a message is sent at rate $R$ to node Y.

\begin{figure}[h]
\psfrag{l1}[][][.8]{$X \sim p_0(x)$}
\psfrag{l2}[][][.8]{$Y$}
\psfrag{l4}[][][.8]{$R$}
\psfrag{l6}[][][.7]{Node X}
\psfrag{l7}[][][.7]{Node Y}
\centerline{\includegraphics[width=.25\textwidth]{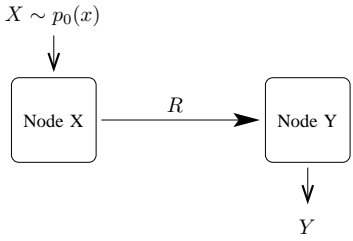}}
\caption{{\em Two nodes.}  The action $X$ is chosen by nature according to $p_0(x)$.  A message is sent to node Y at rate $R$.  The coordination capacity region ${\cal C}_{p_0}$ is the set of rate-coordination pairs where the rate is greater than the mutual information between $X$ and $Y$.}
\label{figure two nodes}
\end{figure}

The $(2^{nR},n)$ coordination codes consist of an encoding function
\begin{eqnarray*}
i & : & {\cal X}^n \longrightarrow \{1,...,2^{nR}\},
\end{eqnarray*}
and a decoding function
\begin{eqnarray*}
y^n & : & \{1,...,2^{nR}\} \longrightarrow {\cal Y}^n.
\end{eqnarray*}

The actions $X^n$ are chosen by nature i.i.d. according to $p_0(x)$, and the actions $Y^n$ are functions of $X^n$ given by implementing the coordination code as
\begin{eqnarray*}
Y^n & = & y^n(i(X^n)).
\end{eqnarray*}

\begin{theorem}[Coordination capacity region]
\label{theorem two nodes}
The coordination capacity region ${\cal C}_{p_0}$ for empirical coordination in the two-node network of Figure \ref{figure two nodes} is the set of rate-coordination pairs where the rate is greater than the mutual information between $X$ and $Y$.  Thus,
\begin{eqnarray*}
{\cal C}_{p_0} & = & \left\{ (R,p(y|x)) \; : \;
\begin{array}{l}
R \geq I(X;Y)
\end{array}
\right\}.
\end{eqnarray*}
\end{theorem}

{\em Discussion:}
The coordination capacity region in this setting yields the rate-distortion result of Shannon \cite{Shannon60}.  Notice that with no communication ($R=0$), only independent distributions $p_0(x)p(y)$ are achievable, in contrast to the setting of Figure \ref{f_all_dist}, where none of the actions were specified by nature and all joint distributions were achievable.

\begin{example}[Task assignment]
\label{example two nodes}
Suppose there are $k$ tasks numbered $1$ through $k$.  One task is dealt randomly to node X, and node Y needs to choose one of the remaining tasks.  This coordinated behavior can be summarized by a distribution $\hat{p}$.  The action $X$ is given by nature according to $\hat{p}_0(x)$, the uniform distribution on the set $\{1,...,k\}$.  The desired conditional distribution of the action $Y$ is $\hat{p}(y|x)$, the uniform distribution on the set of tasks different from $x$.  Therefore, the joint distribution $\hat{p}_0(x)\hat{p}(y|x)$ is the uniform distribution on pairs of differing tasks from the set $\{1,...,k\}$.  Figure \ref{figure two nodes three card} illustrates a valid outcome for $k$ larger than $5$.

\begin{figure}[h]
\psfrag{l1}[][][0.8]{$X \;$}
\psfrag{l2}[][][0.8]{$Y$}
\psfrag{l4}[][][0.8]{}
\psfrag{l5}[][][0.8]{$R$}
\psfrag{l7}[][][0.8]{$k \geq 5:$}
\psfrag{c1}[][][0.8]{$5$}
\psfrag{c2}[][][0.8]{$3$}
\centerline{\includegraphics[width=.25\textwidth]{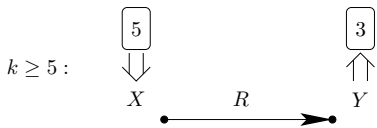}}
\caption{{\em Task assignment in the two-node network.}  A task from a set of tasks numbered $1,...,k$ is to be assigned uniquely to each of the nodes X and Y in the two-node network setting.  The task assignment for X is given randomly by nature.  The communication rate $R \geq \log (k/k-1)$ is necessary and sufficient to allow Y to select a different task from X.}
\label{figure two nodes three card}
\end{figure}

By applying Theorem \ref{theorem two nodes}, we find that the rate-coordination region ${\cal R}_{\hat{p}_0}(\hat{p}(y|x))$ is given by
\begin{eqnarray*}
{\cal R}_{\hat{p}_0}(\hat{p}(y|x)) & = & \left\{ R \; : \; R \geq \log \left( \frac{k}{k-1} \right) \right\}.
\end{eqnarray*}
\end{example}

\subsection{Isolated node}
\label{subsection isolated}

Now we derive the coordination capacity region for the isolated-node network of Figure \ref{figure isolated}.  Node X has an action chosen by nature according to $p_0(x)$, and a message is sent at rate $R$ from node X to node Y from which node Y produces an action.  Node Z also produces an action but receives no communication.  What is the set of all achievable coordination distributions $p(y,z|x)$?  At first it seems that the action at the isolated node Z must be independent of $Y$, but we will see otherwise.

\begin{figure}[h]
\psfrag{l1}[][][.8]{$X \sim p_0(x)$}
\psfrag{l2}[][][.8]{$Y$}
\psfrag{l3}[][][.8]{$Z$}
\psfrag{l4}[][][.8]{$R$}
\psfrag{l5}[][][.8]{$R_2$}
\psfrag{l6}[][][.7]{Node X}
\psfrag{l7}[][][.7]{Node Y}
\psfrag{l8}[][][.7]{Node Z}
\centerline{\includegraphics[width=.25\textwidth]{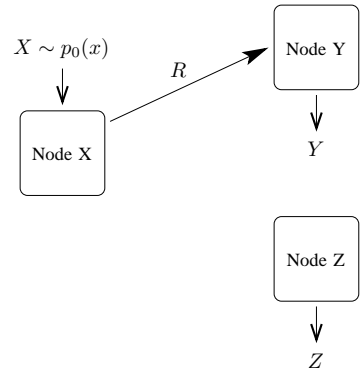}}
\caption{{\em Isolated node.}  The action $X$ is chosen by nature according to $p_0(x)$, and a message is sent at rate $R$ from node X to node Y.  Node Z receives no communication.  The coordination capacity region ${\cal C}_{p_0}$ is the set of rate-coordination pairs where $p(x,y,z) = p_0(x)p(z)p(y|x,z)$ and the rate $R$ is greater than the conditional mutual information between $X$ and $Y$ given $Z$.}
\label{figure isolated}
\end{figure}

We formalize this problem as follows.  The $(2^{nR},n)$ coordination codes consist of an encoding function
\begin{eqnarray*}
i & : & {\cal X}^n \longrightarrow \{1,...,2^{nR}\},
\end{eqnarray*}
a decoding function
\begin{eqnarray*}
y^n & : & \{1,...,2^{nR}\} \longrightarrow {\cal Y}^n,
\end{eqnarray*}
and a deterministic sequence
\begin{eqnarray*}
z^n & \in & {\cal Z}^n.
\end{eqnarray*}

The actions $X^n$ are chosen by nature i.i.d. according to $p_0(x)$, and the actions $Y^n$ are functions of $X^n$ given by implementing the coordination code as
\begin{eqnarray*}
Y^n & = & y^n(i(X^n)), \\
Z^n & = & z^n.
\end{eqnarray*}

The coordination capacity region for this network is given in the following theorem.  As we previously alluded, notice that the action $Z$ need not be independent of $Y$, even though there is no communication to node Z.

\begin{theorem}[Coordination capacity region]
\label{theorem isolated}
The coordination capacity region ${\cal C}_{p_0}$ for empirical coordination in the isolated-node network of Figure \ref{figure isolated} is the set of rate-coordination pairs where $Z$ is independent of $X$ and the rate $R$ is greater than the conditional mutual information between $X$ and $Y$ given $Z$.  Thus,
\begin{eqnarray*}
{\cal C}_{p_0} & = & \left\{ (R,p(z)p(y|x,z)) \; : \;
\begin{array}{l}
R \geq I(X;Y|Z)
\end{array}
\right\}.
\end{eqnarray*}
\end{theorem}

{\em Discussion:}
How can $Y$ and $Z$ have a dependence when there is no communication between them?  This dependence is possible because neither $Y$ nor $Z$ is chosen randomly by nature.  In an extreme case, we could let node Y ignore the incoming message from node X and let the actions at node Y and node Z be equal, $Y=Z$.  Thus we can immediately see that with no communication the coordination region consists of all distributions of the form $p_0(x)p(y,z)$.

If we were to use common randomness $\omega$ to generate the action sequence $Z^n(\omega)$, then Node Y, which also has access to the common randomness, can use it to produce correlated actions.  This does not increase the coordination capacity region (see Theorem \ref{theorem no common randomness}), but it provides an intuitive understanding of how $Y$ and $Z$ can be correlated.  Without explicit use of common randomness, we select a determinist sequence $z^n$ before-hand as part of our codebook and make it known to all parties.

It is interesting to note that there is a tension between the correlation of $X$ and $Y$ and the correlation of $Y$ and $Z$.  For instance, if the communication is used to make perfect correlation between $X$ and $Y$ then any potential correlation between $Y$ and $Z$ is forfeited.

Within the results for the more general cascade network in the sequel (Section \ref{subsection cascade}) we will find that Theorem \ref{theorem isolated} is an immediate consequence of Theorem \ref{theorem cascade} by letting $R_2 = 0$.

\begin{example}[Jointly Gaussian]
\label{example jointly gaussian}
Jointly Gaussian distributions illustrate the tradeoff between the correlation of $X$ and $Y$ and the correlation of $Y$ and $Z$ in the isolated-node network.  Consider the portion of the coordination-rate region ${\cal P}_{p_0}(R)$ that consists of jointly Gaussian distributions.  If $X$ is distributed according to $N(0,\sigma_X^2)$, what set of covariance matrices can be achieved at rate $R$?

So far we have discussed coordination for distribution functions with finite alphabets.  Extending to infinite alphabet distributions, achievability means that any finite quantization of the joint distribution is achievable.

Using Theorem \ref{theorem isolated}, we bound the correlations as follows:
\begin{eqnarray}
\label{equation gaussian example}
R & \geq & I(X;Y|Z) \nonumber \\
& = & I(X;Y,Z) \nonumber \\
& = & \frac{1}{2} \log \frac{|K_x||K_{yz}|}{|K_{XYZ}|} \nonumber \\
& \stackrel{(a)}{=} & \frac{1}{2} \log \frac{\sigma_x^2 (\sigma_y^2 \sigma_z^2 - \sigma_{yz}^2)}{\sigma_x^2\sigma_y^2\sigma_z^2 - \sigma_x^2\sigma_{yz}^2 - \sigma_z^2\sigma_{xy}^2} \nonumber \\
& \stackrel{(b)}{=} & \frac{1}{2} \log \frac{1 - \left(\frac{\sigma_{yz}}{\sigma_y\sigma_z}\right)^2}{1 - \left(\frac{\sigma_{yz}}{\sigma_y\sigma_z}\right)^2 - \left(\frac{\sigma_{xy}}{\sigma_x\sigma_y}\right)^2} \nonumber \\
& = & \frac{1}{2} \log \frac{1-\rho_{yz}^2}{1-\rho_{yz}^2-\rho_{xy}^2},
\end{eqnarray}
where $\rho_{xy}$ and $\rho_{yz}$ are correlation coefficients.  Equality (a) holds because $\sigma_{xz}=0$ due to the independence between $X$ and $Z$.  Obtain equality (b) by dividing the numerator and denominator of the argument of the $\log$ by $\sigma_x^2\sigma_y^2\sigma_z^2$.

Unfolding (\ref{equation gaussian example}) yields a linear tradeoff between the $\rho_{xy}^2$ and $\rho_{yz}^2$, given by
\begin{eqnarray*}
(1-2^{-2R})^{-1} \rho_{xy}^2 + \rho_{yz}^2 & \leq & 1.
\end{eqnarray*}
Thus all correlation coefficients $\rho_{xy}$ and $\rho_{yz}$ satisfying this constraint are achievable at rate $R$.
\end{example}

\subsection{Cascade}
\label{subsection cascade}

We now give the coordination capacity region for the cascade of communication in Figure \ref{figure cascade}.  In this setting, the action at node X is chosen by nature.  A message at rate $R_1$ is sent from node X to node Y, and subsequently a message at rate $R_2$ is sent from node Y to node Z based on the message received from node X.  Nodes Y and Z produce actions based on the messages they receive.

\begin{figure}[h]
\psfrag{l1}[][][.8]{$X \sim p_0(x)$}
\psfrag{l2}[][][.8]{$Y$}
\psfrag{l3}[][][.8]{$Z$}
\psfrag{l4}[][][.8]{$R_1$}
\psfrag{l5}[][][.8]{$R_2$}
\psfrag{l6}[][][.7]{Node X}
\psfrag{l7}[][][.7]{Node Y}
\psfrag{l8}[][][.7]{Node Z}
\centerline{\includegraphics[width=.4\textwidth]{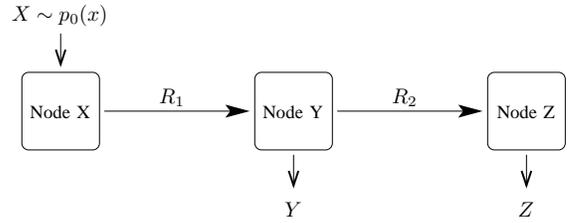}}
\caption{{\em Cascade.}  The action $X$ is chosen by nature according to $p_0(x)$.  A message is sent from node X to node Y at rate $R_1$.  Node Y produces an action $Y$ and a message to send to node Z based on the message received from node X.  Node Z then produces an action $Z$ based on the message received from node Y.  The coordination capacity region ${\cal C}_{p_0}$ is the set of rate-coordination triples where the rate $R_1$ is greater than the mutual information between $X$ and $(Y,Z)$, and the rate $R_2$ is greater than the mutual information between $X$ and $Z$.}
\label{figure cascade}
\end{figure}

The formal statement is as follows.  The $(2^{nR_1},2^{nR_2},n)$ coordination codes consist of four functions---an encoding function
\begin{eqnarray*}
i & : & {\cal X}^n \longrightarrow \{1,...,2^{nR_1}\},
\end{eqnarray*}
a recoding function
\begin{eqnarray*}
j & : & \{1,...,2^{nR_1}\} \longrightarrow \{1,...,2^{nR_2}\},
\end{eqnarray*}
and two decoding functions
\begin{eqnarray*}
y^n & : & \{1,...,2^{nR_1}\} \longrightarrow {\cal Y}^n, \\
z^n & : & \{1,...,2^{nR_2}\} \longrightarrow {\cal Z}^n.
\end{eqnarray*}

The actions $X^n$ are chosen by nature i.i.d. according to $p_0(x)$, and the actions $Y^n$ and $Z^n$ are functions of $X^n$ given by implementing the coordination code as
\begin{eqnarray*}
Y^n & = & y^n(i(X^n)), \\
Z^n & = & z^n(j(i(X^n))).
\end{eqnarray*}

This network was considered by Yamamoto \cite{Yamamoto_source_coding81} in the context of rate-distortion theory.  The same optimal encoding scheme from his work achieves the coordination capacity region as well.

\begin{theorem}[Coordination capacity region]
\label{theorem cascade}
The coordination capacity region ${\cal C}_{p_0}$ for empirical coordination in the cascade network of Figure \ref{figure cascade} is the set of rate-coordination triples where the rate $R_1$ is greater than the mutual information between $X$ and $(Y,Z)$, and the rate $R_2$ is greater than the mutual information between $X$ and $Z$.  Thus,
\begin{eqnarray*}
{\cal C}_{p_0} & = & \left\{ (R_1,R_2,p(y,z|x)) \; : \;
\begin{array}{l}
R_1 \geq I(X;Y,Z), \\
R_2 \geq I(X;Z).
\end{array}
\right\}.
\end{eqnarray*}
\end{theorem}

{\em Discussion:}
The coordination capacity region ${\cal C}_{p_0}$ meets the cut-set bound.  The trick to achieving this bound is to first specify $Z$ and then specify $Y$ conditioned on $Z$.

\begin{example}[Task assignment]
\label{example cascade}
Consider a task assignment setting where three tasks are to be assigned without duplication to the three nodes X, Y, and Z, and the assignment for node X is chosen uniformly at random by nature.  A distribution capturing this coordination behavior is the uniform distribution over the six permutations of task assignments.  Let $\hat{p}_0(x)$ be the uniform distribution on the set $\{1,2,3\}$, and let $\hat{p}(y,z|x)$ give equal probability to both of the assignments to Y and Z that produce different tasks at the three nodes.  Figure \ref{figure cascade three card} illustrates a valid outcome of the task assignments.

\begin{figure}[h]
\psfrag{l1}[][][0.8]{$X \;$}
\psfrag{l2}[][][0.8]{$Y \;$}
\psfrag{l3}[][][0.8]{$Z \;$}
\psfrag{l4}[][][0.8]{}
\psfrag{l5}[][][0.8]{$R_1$}
\psfrag{l6}[][][0.8]{$R_2$}
\psfrag{c1}[][][0.8]{$2$}
\psfrag{c2}[][][0.8]{$1$}
\psfrag{c3}[][][0.8]{$3$}
\centerline{\includegraphics[width=6cm]{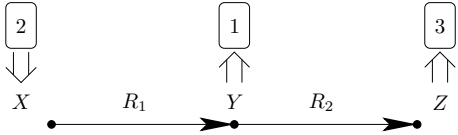}}
\caption{{\em Task assignment in the cascade network.}  Three tasks, numbered $1$, $2$, and $3$, are distributed among three nodes X, Y, and Z in the cascade network setting.  The task assignment for X is given randomly by nature.  The rates $R_1 \geq \log 3$ and $R_2 \geq \log 3 - \log 2$ are required to allow Y and Z to choose different tasks from X and from each other.}
\label{figure cascade three card}
\end{figure}

According to Theorem \ref{theorem cascade}, the rate-coordination region ${\cal R}_{\hat{p}_0}(\hat{p}(y,z|x))$ is given by
\begin{eqnarray*}
{\cal R}_{\hat{p}_0}(\hat{p}(y,z|x)) & = & \left\{ (R_1,R_2) \; : \;
\begin{array}{l}
R_1 \geq \log 3, \\
R_2 \geq \log 3 - \log 2.
\end{array}
\right\}.
\end{eqnarray*}
\end{example}

\subsection{Degraded source}
\label{subsection degraded}

Here we present the coordination capacity region for the degraded-source network shown in Figure \ref{figure degraded source}.  Nodes X and Y each have an action specified by nature, and $Y$ is a function of $X$.  That is, $p_0(x,y) = p_0(x) {\bf 1} (y = f_0(x))$, where ${\bf 1}(\cdot)$ is the indicator function.  Node X sends a message to node Y at rate $R_1$ and a message to node Z at rate $R_2$.  Node Y, upon receiving the message from node X, sends a message at rate $R_3$ to node Z.  Node Z produces an action based on the two messages it receives.

\begin{figure}[h]
\psfrag{l1}[][][.8]{$X \sim p_0(x)$}
\psfrag{l2}[][][.8]{$Y = f_0(X)$}
\psfrag{l3}[][][.8]{$Z$}
\psfrag{l4}[][][.8]{$R_1$}
\psfrag{l5}[][][.8]{$R_3$}
\psfrag{l6}[][][.8]{$R_2$}
\psfrag{l7}[][][.7]{Node X}
\psfrag{l8}[][][.7]{Node Y}
\psfrag{l9}[][][.7]{Node Z}
\centerline{\includegraphics[width=.4\textwidth]{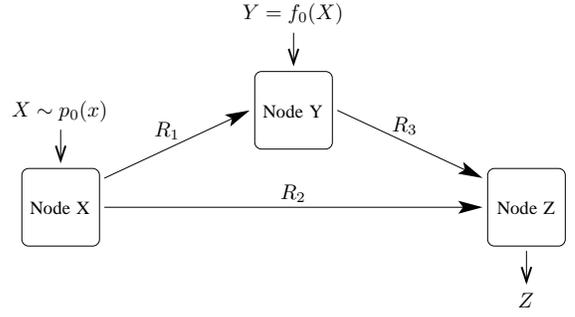}}
\caption{{\em Degraded source:}  The action $X$ is specified by nature according to $p_0(x)$, and the action $Y$ is a function $f_0$ of $X$.  A message is sent from node X to node Y at rate $R_1$, after which node Y constructs a message for node Z at rate $R_3$ based on the incoming message from node X and the action $Y$.  Node X also sends a message directly to node Z at rate $R_2$.  The coordination capacity region ${\cal C}_{p_0}$ is given in Theorem \ref{theorem degraded source}.}
\label{figure degraded source}
\end{figure}

The $(2^{nR_1},2^{nR_2},2^{nR_3},n)$ coordination codes for Figure \ref{figure degraded source} consist of four functions---two encoding functions
\begin{eqnarray*}
i & : & {\cal X}^n \longrightarrow \{1,...,2^{nR_1}\}, \\
j & : & {\cal X}^n \longrightarrow \{1,...,2^{nR_2}\},
\end{eqnarray*}
a recoding function
\begin{eqnarray*}
k & : & \{1,...,2^{nR_1}\} \times {\cal Y}^n \longrightarrow \{1,...,2^{nR_3}\},
\end{eqnarray*}
and a decoding function
\begin{eqnarray*}
z^n & : & \{1,...,2^{nR_2}\} \times \{1,...,2^{nR_3}\} \longrightarrow {\cal Y}^n.
\end{eqnarray*}

The actions $X^n$ and $Y^n$ are chosen by nature i.i.d. according to $p_0(x,y)$, having the property that $Y_i = f_0(X_i)$ for all $i$, and the actions $Z^n$ are a function of $X^n$ and $Y^n$ given by implementing the coordination code as
\begin{eqnarray*}
Y^n & = & y^n(j(X^n),k(i(X^n),Y^n)).
\end{eqnarray*}

Others have investigated source coding networks in the rate-distortion context where two sources are encoded at separate nodes to be reconstructed at a third node.  Kaspi and Berger \cite{Kaspi_berger82} consider a variety of cases where the encoders share some information.  Also, Barros and Servetto \cite{barros-servetto04} articulate the compress and bin strategy for more general bi-directional exchanges of information among the encoders.  While falling under the same general compression strategy, the degraded source network is a special case where optimality can be established, yielding a characterization of the coordination capacity region.

\begin{theorem}[Coordination capacity region]
\label{theorem degraded source}
The coordination capacity region ${\cal C}_{p_0}$ for empirical coordination in the degraded-source network of Figure \ref{figure degraded source} is given by
\begin{eqnarray*}
{\cal C}_{p_0} = \left\{ (R_1,R_2,R_3,p(z|x,y)) \; : \;
\begin{array}{l}
\exists p(u|x,y,z) \mbox{ such that} \\
|{\cal U}| \leq |{\cal X}||{\cal Z}| + 2, \\
R_1 \geq I(X;U|Y), \\
R_2 \geq I(X;Z|U), \\
R_3 \geq I(X;U).
\end{array}
\right\}.
\end{eqnarray*}
\end{theorem}

%% file: coordination_partial_results.tex
\section{Coordination---Partial Results}
\label{section partial results}

We have given the coordination capacity region for several multinode networks.  Those results are complete.  We now investigate networks for which we have only partial results.

In this section we present bounds on the coordination capacity regions ${\cal C}_{p_0}$ for empirical coordination in two network settings of three nodes---the broadcast network and the cascade-multiterminal network.  A communication technique that we find useful in both settings, also used in the degraded-source network of Section \ref{section complete results}, is to use a portion of the communication to send identical messages to all nodes in the network.  The common message serves to correlate the codebooks used on different communication links and can result in reduced rates in the network.

Proofs are left to Section \ref{section proofs}.  Again, as a consequence of Theorem \ref{theorem no common randomness} we need not use common randomness in this section.

\subsection{Broadcast}
\label{subsection broadcast}

We now give bounds on the coordination capacity region for the broadcast network of Figure \ref{figure broadcast}.  In this setting, node X has an action specified by nature according to $p_0(x)$ and sends one message to node Y at rate $R_1$ and a separate message to node Z at rate $R_2$.  Nodes Y and Z each produce an action based on the message they receive.

\begin{figure}[h]
\psfrag{l1}[][][.8]{$X \sim p_0(x)$}
\psfrag{l2}[][][.8]{$Y$}
\psfrag{l3}[][][.8]{$Z$}
\psfrag{l4}[][][.8]{$R_1$}
\psfrag{l5}[][][.8]{$R_2$}
\psfrag{l6}[][][.7]{Node X}
\psfrag{l7}[][][.7]{Node Y}
\psfrag{l8}[][][.7]{Node Z}
\centerline{\includegraphics[width=.25\textwidth]{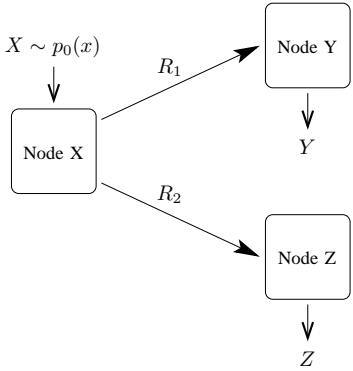}}
\caption{{\em Broadcast.}  The action $X$ is chosen by nature according to $p_0(x)$.  A message is sent from node X to node Y at rate $R_1$, and a separate message is sent from node X to node Z at rate $R_2$.  Nodes Y and Z produce actions based on the messages they receive.  Bounds on the coordination capacity region ${\cal C}_{p_0}$ are given in Theorem \ref{theorem broadcast}.}
\label{figure broadcast}
\end{figure}

Node X serves as the controller for the network.  Nature assigns an action to node X, which then tells node Y and node Z which actions to take.

The $(2^{nR_1},2^{nR_2},n)$ coordination codes consist of two encoding functions
\begin{eqnarray*}
i & : & {\cal X}^n \longrightarrow \{1,...,2^{nR_1}\}, \\
j & : & {\cal X}^n \longrightarrow \{1,...,2^{nR_2}\},
\end{eqnarray*}
and two decoding functions
\begin{eqnarray*}
y^n & : & \{1,...,2^{nR_1}\} \longrightarrow {\cal Y}^n. \\
z^n & : & \{1,...,2^{nR_2}\} \longrightarrow {\cal Z}^n.
\end{eqnarray*}

The actions $X^n$ are chosen by nature i.i.d. according to $p_0(x)$, and the actions $Y^n$ and $Z^n$ are functions of $X^n$ given by implementing the coordination code as
\begin{eqnarray*}
Y^n & = & y^n(i(X^n)). \\
Z^n & = & z^n(j(X^n)).
\end{eqnarray*}

From a rate-distortion point of view, the broadcast network is not a likely candidate for consideration.  The problem separates into two non-interfering rate-distortion problems, and the relationship between the sequences $Y^n$ and $Z^n$ is ignored (unless the decoders communicate as in \cite{Yamamoto_triangle96}).  However, a related scenario, the problem of multiple descriptions \cite{WolfWynerZiv}, where the combination of two messages $I$ and $J$ are used to make a third estimate of the source $X$, demands consideration of the relationship between the two messages.  In fact, the communication scheme for the multiple descriptions problem presented by Zhang and Berger \cite{zhang87Berger} coincides with our inner bound for the coordination capacity region in the broadcast network.

The set of rate-coordination tuples ${\cal C}_{p_0,in}$ is an inner bound on the coordination capacity region, given by
\begin{eqnarray*}
{\cal C}_{p_0,in} & \triangleq & \phantom{extra space to fill the line with}
\end{eqnarray*}
\begin{eqnarray*}
\left\{
\begin{array}{l}
(R_1,R_2,p(y,z|x)) \; : \;
\exists p(u|x,y,z) \mbox{ such that} \\
R_1 \geq I(X;U,Y), \\
R_2 \geq I(X;U,Z), \\
R_1 + R_2 \geq I(X;U,Y) + I(X;U,Z) + I(Y;Z|X,U).
\end{array}
\right\}. & &
\end{eqnarray*}
The set of rate-coordination tuples ${\cal C}_{p_0,out}$ is an outer bound on the coordination capacity region, given by
\begin{eqnarray*}
{\cal C}_{p_0,out} & \triangleq & \left\{
\begin{array}{l}
(R_1,R_2,p(y,z|x)) \; : \; \\
R_1 \geq I(X;Y), \\
R_2 \geq I(X;Z), \\
R_1 + R_2 \geq I(X;Y,Z).
\end{array}
\right\}.
\end{eqnarray*}
Also, define ${\cal R}_{p_0,in}(p(y,z|x))$ and ${\cal R}_{p_0,out}(p(y,z|x))$ to be the sets of rate pairs in ${\cal C}_{p_0,in}$ and ${\cal C}_{p_0,out}$ corresponding to the desired distribution $p(y,z|x)$.

\begin{theorem}[Coordination capacity region bounds]
\label{theorem broadcast}
The coordination capacity region ${\cal C}_{p_0}$ for empirical coordination in the broadcast network of Figure \ref{figure broadcast} is bounded by
\begin{eqnarray*}
{\cal C}_{p_0,in} & \subset & {\cal C}_{p_0} \;\; \subset \;\; {\cal C}_{p_0,out}.
\end{eqnarray*}
\end{theorem}

{\em Discussion:}
The regions ${\cal C}_{p_0,in}$ and ${\cal C}_{p_0,out}$ are convex.  A time-sharing random variable can be lumped into the auxiliary random variable $U$ in the definition of ${\cal C}_{p_0,in}$ to show convexity.

The inner bound ${\cal C}_{p_0,in}$ is achieved by first sending a common message, represented by $U$, to both receivers and then private messages to each.  The common message effectively correlates the two codebooks to reduce the required rates for specifying the actions $Y^n$ and $Z^n$.  The sum rate takes a penalty of $I(Y;Z|X,U)$ in order to assure that $Y$ and $Z$ are coordinated with each other as well as with $X$.

The outer bound ${\cal C}_{p_0,out}$ is a consequence of applying the two-node result of Theorem \ref{theorem two nodes} in three different ways, once for each receiver, and once for the pair of receivers with full cooperation.

For many distributions, the bounds in Theorem \ref{theorem broadcast} are tight and the rate-coordination region ${\cal R}_{p_0} = {\cal R}_{p_0,in} = {\cal R}_{p_0,out}$.  This is true for all distributions where $X$, $Y$, and $Z$ form a Markov chain in any order.  It is also true for distributions where $Y$ and $Z$ are independent or where $X$ is independent pairwise with both $Y$ and $Z$.  For each of these cases, Table \ref{table broadcast} shows the choice of auxiliary random variable $U$ in the definition of ${\cal R}_{p_0,in}$ that yields ${\cal R}_{p_0,in} = {\cal R}_{p_0,out}$.  In case 5, the region ${\cal R}_{p_0,in}$ is optimized by time-sharing between $U=Y$ and $U=Z$.

\begin{table}[h]
\caption{Known capacity region (cases where ${\cal R}_{p_0,in} = {\cal R}_{p_0,out}$).}
\label{table broadcast}
\centering
\begin{tabular}{|l|c|l|}
\hline
& Condition & Auxiliary \\
\hline
Case 1: & $Y-X-Z$ & $U=\emptyset$ \\
Case 2: & $X-Y-Z$ & $U=Z$ \\
Case 3: & $X-Z-Y$ & $U=Y$ \\
Case 4: & $Y \perp Z$ & $U=\emptyset$ \\
Case 5: & $X \perp Y$ and $X \perp Z$ & $U=Y, U=Z$ \\
\hline
\end{tabular}
\end{table}

Notice that if $R_2=0$ in the broadcast network we find ourselves in the isolated node setting of Section \ref{subsection isolated}.  Consider a particular distribution $p_0(x)p(z)p(y|x,z)$ that could be achieved in the isolated node network.  In the setting of the broadcast network, it might seem that the message from node X to node Z is useless for achieving $p_0(x)p(z)p(y|x,z)$, since $X$ and $Z$ are independent.  However, this is not the case.  For some desired distributions $p_0(x)p(z)p(y|x,z)$, a positive rate $R_2$ in the broadcast network actually helps reduce the required rate $R_1$.

To highlight a specific case where a message to node Z is useful even though $Z$ is independent of $X$ in the desired distribution, consider the following.  Let $\overline{p}_0(x)\overline{p}(z)\overline{p}(y|x,z)$ be the uniform distribution over all combinations of binary $x$, $y$, and $z$ with even parity.  The variables $X$, $Y$, and $Z$ are each Bernoulli-half and pairwise independent, and $X \oplus Y \oplus Z = 0$, where $\oplus$ is addition modulo two.  This distribution satisfies both case 4 and case 5 from Table \ref{table broadcast}, so we know that ${\cal R}_{\overline{p}_0} = {\cal R}_{\overline{p}_0,out}$.  Therefore, the rate-coordination region ${\cal R}_{\overline{p}_0}(\overline{p}(y,z|x))$ is characterized by a single inequality,
\begin{eqnarray*}
{\cal R}_{\overline{p}_0}(\overline{p}(y,z|x)) & = & \{ (R_1,R_2) \in \mathbb{R}_+^2 \; : \; R_1 + R_2 \geq 1 \mbox{ bit} \}.
\end{eqnarray*}
The minimum rate $R_1$ needed when no message is sent from node X to node Z is 1 bit, while the required rate in general is $1 - R_2$ bits.

The following task assignment problem has practical importance.

\begin{example}[Task assignment]
\label{example broadcast}
Consider a task assignment setting similar to Example \ref{example cascade}, where three tasks are to be assigned without duplication to the three nodes X, Y, and Z, and the assignment for node X is chosen uniformly at random by nature.  A distribution capturing this coordination behavior is the uniform distribution over the six permutations of task assignments.  Let $\hat{p}_0(x)$ be the uniform distribution on the set $\{0,1,2\}$, and let $\hat{p}(y,z|x)$ give equal probability to both of the assignments to Y and Z that produce different tasks at the three nodes.  Figure \ref{figure broadcast three card} illustrates a valid outcome of the task assignments.

\begin{figure}[h]
\psfrag{l1}[][][0.8]{$\;\;\;\;\; X$}
\psfrag{l2}[][][0.8]{$Y \;\;\;\;\;$}
\psfrag{l3}[][][0.8]{$Z \;\;\;\;\;$}
\psfrag{l4}[][][0.8]{}
\psfrag{l5}[][][0.8]{$R_1$}
\psfrag{l6}[][][0.8]{$R_2$}
\psfrag{c1}[][][0.8]{$1$}
\psfrag{c2}[][][0.8]{$0$}
\psfrag{c3}[][][0.8]{$2$}
\centerline{\includegraphics[width=6cm]{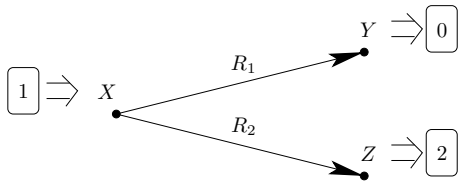}}
\caption{{\em Task assignment in the broadcast network.}  Three tasks, numbered $0$, $1$, and $2$, are distributed among three nodes X, Y, and Z in the broadcast network setting.  The task assignment for X is given randomly by nature.  What rates $R_1$ and $R_2$ are necessary to allow Y and Z to choose different tasks from X and each other?}
\label{figure broadcast three card}
\end{figure}

We can explore the achievable rate region ${\cal R}_{\hat{p}_0}(\hat{p}(y,z|x))$ by using the bounds in Theorem \ref{theorem broadcast}.  In this process, we find rates as low as $\log 3 - \log \phi$ to be sufficient on each link, where $\phi = \frac{\sqrt{5}+1}{2}$ is the golden ratio.

\begin{figure}[h]
\psfrag{l1}[][][0.8]{A}
\psfrag{l2}[][][0.8]{B}
\psfrag{l3}[][][0.8]{C}
\psfrag{l4}[][][0.8]{D}
\psfrag{l5}[][][0.8]{$R_1$}
\psfrag{l6}[][][0.8]{$R_2$}
\psfrag{l7}[][][0.7]{$1$}
\psfrag{l8}[][][0.7]{$\; 1$}
\centerline{\includegraphics[width=5cm]{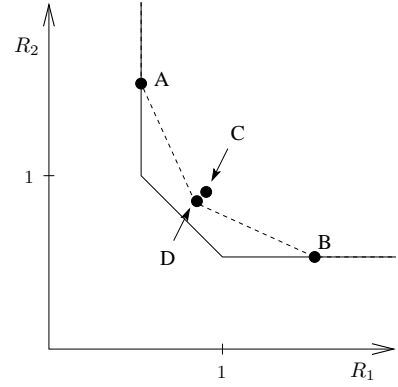}}
\caption{{\em Rate region bounds for task assignment.}  Points $A$, $B$, $C$, and $D$ are achievable rates for the task assignment problem in the broadcast network.  The solid line indicates the outer bound ${\cal R}_{\hat{p}_0,out}(\hat{p}(y,z|x))$, and the dashed line indicates a subset of the inner bound ${\cal R}_{\hat{p}_0,in}(\hat{p}(y,z|x))$.  Points $A$ and $B$ are achieved by letting $U = \emptyset$.  Point $C$ uses $U$ as time-sharing, independent of $X$.  Point $D$ uses $U$ to describe $X$ partially to each of the nodes Y and Z.}
\label{figure three card rates}
\end{figure}

First consider the points in the inner bound ${\cal R}_{\hat{p}_0, in}(\hat{p}(y,z|x))$ that are achieved without the use of the auxiliary variable $U$.  This consists of a pentagonal region of rate pairs.  The extreme point $A = (\log (3/2), \log 3)$, shown in Figure \ref{figure three card rates}, corresponds to the a simple communication approach.  First node X coordinates with node Y.  Theorem \ref{theorem two nodes} for the two-node network declares the minimum rate needed to be $R_1 = \log (3/2)$.  After action $Y$ has been established, node X specifies action $Z$ in it's entire detail using the rate $R_2 = \log 3$.  A complementary scheme achieves the extreme point $B$ in Figure \ref{figure three card rates}.  The sum rate achieved by these points is $R_1 + R_2 = 2 (\log_2 3 - 1/2)$ bits.

We can explore more of the inner bound ${\cal R}_{\hat{p}_0, in}(\hat{p}(y,z|x))$ by adding the element of time-sharing.  That is, use an auxiliary variable $U$ that is independent of $X$.  As long as we can assign tasks in the network so that $X$, $Y$, and $Z$ are each unique, then there will be a method of using time-sharing that will achieve the desired uniform distribution over unique task assignments $\hat{p}$.  For example, devise six task assignment schemes from the one successful scheme by mapping the tasks onto the six different permutations of $\{0,1,2\}$.  By time-sharing equally among these six schemes, we achieve the desired distribution.

With the idea of time-sharing in mind, we achieve a better sum rate by restricting the domain of $Y$ to $\{0,1\}$ and $Z$ to $\{0,2\}$ and letting them be functions of $X$ in the following way:
\begin{eqnarray}
Y & = & \left\{
\begin{array}{ll}
1, & X \neq 1, \\
0, & X = 1,
\end{array}
\right. \\
Z & = & \left\{
\begin{array}{ll}
2, & X \neq 2, \\
0, & X = 2.
\end{array}
\right.
\end{eqnarray}
We can say that $Y$ takes on a default value of $1$, and $Z$ takes on a default value of $2$.  Node X just tells nodes Y and Z when they need to get out of the way, in which case they switch to task $0$.  To achieve this we only need $R_1 \geq H(Y) = \log_3 - 2/3$ bits and $R_2 \geq H(Z) = \log_2 3 - 2/3$ bits, represented by point $C$ in Figure \ref{figure three card rates}.

Finally, we achieve an even smaller sum rate in the inner bound ${\cal R}_{\hat{p}_0, in}(\hat{p}(y,z|x))$ by using a more interesting choice of $U$ in addition to time-sharing.\footnote{Time-sharing is also lumped into $U$, but we ignore that here to simplify the explanation.}  Let $U \in \{0,1,2\}$ be correlated with $X$ in such a way that they are equal more often than one third of the time.  Now restrict the domains of $Y$ and $Z$ based on $U$.  The actions $Y$ and $Z$ are functions of $X$ and $U$ defined as follows:
\begin{eqnarray}
Y & = & \left\{
\begin{array}{ll}
U + 1 \mbox{ mod 3}, & X \neq U + 1 \mbox{ mod 3}, \\
U, & X = U + 1 \mbox{ mod 3},
\end{array}
\right. \\
Z & = & \left\{
\begin{array}{ll}
U - 1 \mbox{ mod 3}, & X \neq U - 1 \mbox{ mod 3}, \\
U, & X = U - 1 \mbox{ mod 3}.
\end{array}
\right.
\end{eqnarray}
This corresponds to sending a compressed description of $X$, represented by $U$, and then assigning default values to $Y$ and $Z$ centered around $U$.  The actions $Y$ and $Z$ sit on both sides of $U$ and only move when X tells them to get out of the way.  The description rates needed for this method are
\begin{eqnarray}
R_1 & \geq & I(X;U) + I(X;Y|U) \nonumber \\
& = & I(X;U) + H(Y|U). \nonumber \\
R_2 & \geq & I(X;U) + I(X;Z|U) \nonumber \\
& = & I(X;U) + H(Z|U).
\end{eqnarray}

Using a symmetric conditional distribution from $X$ to $U$,  calculus provides the following parameters:
\begin{eqnarray}
P(U=u|X=x) & = & \left\{
\begin{array}{ll}
\frac{1}{\sqrt{5}}, & u = x, \\
\frac{1}{\phi \sqrt{5}}, & u \neq x,
\end{array}
\right. \\
\end{eqnarray}
where $\phi = \frac{\sqrt{5}+1}{2}$ is the golden ratio.  This level of compression results in a very low rate of description, $I(X;U) \approx 0.04$ bits, for sending $U$ to each of the nodes Y and Z.

The description rates needed for this method are as follows, and are represented by Point D in Figure \ref{figure three card rates}:
\begin{eqnarray}
R_1 & \geq & I(X;U) + H(Y|U) \nonumber \\
& = & \log 3 - \frac{1}{2} \log 5 - \frac{2}{\phi \sqrt{5}} \log \phi + H(Y|U) \nonumber \\
& = & \log 3 - \frac{1}{2} \log 5 - \frac{2}{\phi \sqrt{5}} \log \phi + H\left( \frac{1}{\phi \sqrt{5}} \right) \nonumber \\
& = & \log 3 - \frac{2}{\phi \sqrt{5}} \log \phi + \frac{1}{\phi \sqrt{5}} \log \phi - \frac{\phi}{\sqrt{5}} \log \phi \nonumber \\
& = & \log 3 - \left( \phi + \frac{1}{\phi} \right) \frac{1}{\sqrt{5}} \log \phi \nonumber \\
& = & \log 3 - \log \phi, \nonumber \\
R_2 & \geq & \log 3 - \log \phi,
\end{eqnarray}
where $H$ is the binary entropy function.  The above calculation is assisted by observing that $\phi = \frac{1}{\phi}+1$ and $\phi + \frac{1}{\phi} = \sqrt{5}$.
\end{example}

\subsection{Cascade multiterminal}
\label{subsection relay}

We now give bounds on the coordination capacity region for the cascade-multiterminal network of Figure \ref{figure relay}.  In this setting, node X and node Y each have an action specified by nature according to the joint distribution $p_0(x,y)$.  Node X sends a message at rate $R_1$ to node Y.  Based on its own action $Y$ and the incoming message about $X$, node Y sends a message to node Z at rate $R_2$.  Finally, node Z produces an action based on the message from node Y.

\begin{figure}[h]
\psfrag{l1}[][][.8]{$X \sim p_0(x,y)$}
\psfrag{l2}[][][.8]{$Y \sim p_0(x,y)$}
\psfrag{l3}[][][.8]{$Z$}
\psfrag{l4}[][][.8]{$R_1$}
\psfrag{l5}[][][.8]{$R_2$}
\psfrag{l6}[][][.7]{Node X}
\psfrag{l7}[][][.7]{Node Y}
\psfrag{l8}[][][.7]{Node Z}
\centerline{\includegraphics[width=.4\textwidth]{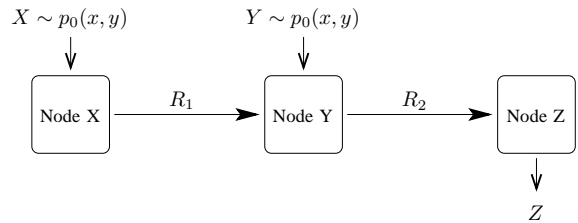}}
\caption{{\em Cascade multiterminal.}  The actions $X$ and $Y$ are chosen by nature according to $p_0(x,y)$.  A message is sent from node X to node Y at rate $R_1$.  Node Y then constructs a message for node Z based on the received message from node X and its own action.  Node Z produces an action based on the message it receives from node Y.  Bounds on the coordination capacity region ${\cal C}_{p_0}$ are given in Theorem \ref{theorem relay}.}
\label{figure relay}
\end{figure}

The $(2^{nR_1},2^{nR_2},n)$ coordination codes consist of an encoding function
\begin{eqnarray*}
i & : & {\cal X}^n \longrightarrow \{1,...,2^{nR_1}\},
\end{eqnarray*}
a recoding function
\begin{eqnarray*}
j & : & \{1,...,2^{nR_1}\} \times {\cal Y}^n \longrightarrow \{1,...,2^{nR_2}\},
\end{eqnarray*}
and a decoding function
\begin{eqnarray*}
z^n & : & \{1,...,2^{nR_2}\} \longrightarrow {\cal Z}^n.
\end{eqnarray*}

The actions $X^n$ and $Y^n$ are chosen by nature i.i.d. according to $p_0(x,y)$, and the actions $Z^n$ are functions of $X^n$ and $Y^n$ given by implementing the coordination code as
\begin{eqnarray*}
Z^n & = & z^n(j(i(X^n),Y^n)).
\end{eqnarray*}

Node Y is playing two roles in this network.  It acts partially as a relay to send on the message from node X to node Z, while at the same time sending a message about its own actions to node Z.  This situation applies to a variety of source coding scenarios.  Nodes X and Y might both be sensors in a sensor network, or node Y can be thought of as a relay for connecting node X to node Z, with side information $Y$.

This network is similar to multiterminal source coding considered by Berger and Tung \cite{Berger78} in that two sources of information are encoded in a distributed fashion.  In fact, the expansion to accommodate cooperative encoders \cite{Kaspi_berger82} can be thought of as a generalization of our network.  However, previous work along these lines is missing one key aspect of efficiency, which is to partially relay the encoded information without changing it.

Vasudevan, Tian, and Diggavi \cite{vasudevan} looked at a similar cascade communication system with a relay. In their setting, the relay's information $Y$ is a degraded version of the decoder's side information, and the decoder is only interested in recovering $X$. Because the relay's observations contain no additional information for the decoder, the relay does not face the dilemma of mixing in some of the side information into its outgoing message. In our cascade multiterminal network, the decoder does not have side information. Thus, the relay is faced with coalescing the two pieces of information $X$ and $Y$ into a single message.  Other research involving similar network settings can be found in \cite{Gu06}, where Gu and Effros consider a more general network but with the restriction that the action $Y$ is a function of the action $X$, and \cite{Bakshi07}, where Bakshi et. al. identify the optimal rate region for lossless encoding of independent sources in a longer cascade (line) network.

The set of rate-coordination tuples ${\cal C}_{p_0,in}$ is an inner bound on the coordination capacity region, given by
\begin{eqnarray*}
{\cal C}_{p_0,in} & \triangleq & \phantom{fill space with a lot of nothing much, so I think}
\end{eqnarray*}
\begin{eqnarray*}
& & \left\{
\begin{array}{l}
(R_1,R_2,p(z|x,y)) \; : \; \\
\exists p(u,v|x,y,z) \mbox{ such that} \\
p(x,y,z,u,v) = p_0(x,y)p(u,v|x)p(z|y,u,v) \\
R_1 \geq I(X;U,V|Y), \\
R_2 \geq I(X;U) + I(Y,V;Z|U).
\end{array}
\right\}.
\end{eqnarray*}
The set of rate-coordination tuples ${\cal C}_{p_0,out}$ is an outer bound on the coordination capacity region, given by
\begin{eqnarray*}
{\cal C}_{p_0,out} & \triangleq & \left\{
\begin{array}{l}
(R_1,R_2,p(z|x,y)) \; : \; \\
\exists p(u|x,y,z) \mbox{ such that} \\
p(x,y,z,u) = p_0(x,y)p(u|x)p(z|y,u) \\
|{\cal U}| \leq |{\cal X}||{\cal Y}||{\cal Z}|, \\
R_1 \geq I(X;U|Y), \\
R_2 \geq I(X,Y;Z).
\end{array}
\right\}.
\end{eqnarray*}
Also, define ${\cal R}_{p_0,in}(p(z|x,y))$ and ${\cal R}_{p_0,out}(p(z|x,y))$ to be the sets of rate pairs in ${\cal C}_{p_0,in}$ and ${\cal C}_{p_0,out}$ corresponding to the desired distribution $p(z|x,y)$.

\begin{theorem}[Coordination capacity region bounds]
\label{theorem relay}
The coordination capacity region ${\cal C}_{p_0}$ for empirical coordination in the cascade multiterminal network of Figure \ref{figure relay} is bounded by
\begin{eqnarray*}
{\cal C}_{p_0,in} & \subset & {\cal C}_{p_0} \;\; \subset \;\; {\cal C}_{p_0,out}.
\end{eqnarray*}
\end{theorem}

{\em Discussion:}
The regions ${\cal C}_{p_0,in}$ and ${\cal C}_{p_0,out}$ are convex.  A time-sharing random variable can be lumped into the auxiliary random variable $U$ in the definition of ${\cal C}_{p_0,in}$ to show convexity.

The inner bound ${\cal C}_{p_0,in}$ is achieved by dividing the message from node X into two parts.  One part, represented by $U$, is sent to all nodes, relayed by node Y to node Z.  The other part, represented by $V$, is sent only to node Y.  Then node Y recompresses $V$ along with $Y$.

The outer bound ${\cal C}_{p_0,out}$ is a combination of the Wyner-Ziv \cite{WynerZiv76} bound for source coding with side information at the decoder, obtained by letting node Y and node Z fully cooperate, and the two-node bound of Theorem \ref{theorem two nodes}, obtained by letting node X and node Y fully cooperate.

For some distributions, the bounds in Theorem \ref{theorem relay} are tight and the rate-coordination region ${\cal R}_{p_0} = {\cal R}_{p_0,in} = {\cal R}_{p_0,out}$.  This is true for all distributions where $X-Y-Z$ form a Morkov chain or $Y-X-Z$ form a Markov chain.  In the first case, where $X-Y-Z$ form a Morkov chain, choosing $U=V=\emptyset$ in the definition of ${\cal C}_{p_0,in}$ reduces the region to all rate pairs such that $R_2 \geq I(Y;Z)$, which meets the outer bound ${\cal C}_{p_0,out}$.  In the second case, where $Y-X-Z$ form a Morkov chain, choosing $U=Z$ and $V=\emptyset$ reduces the region to all rate pairs such that $R_1 \geq I(X;Z|Y)$ and $R_2 \geq I(X;Z)$, which meets the outer bound.  Therefore, we find as special cases that the bounds in Theorem \ref{theorem relay} are tight if $X$ is a function of $Y$, if $Y$ is a function of $X$, or if the reconstruction $Z$ is a function of $X$ and $Y$ \cite{Cuff09}.

Table \ref{table relay} shows choices of $U$ and $V$ from ${\cal R}_{p_0,in}$ that yield ${\cal R}_{p_0,in} = {\cal R}_{p_0,out}$ in each of the above cases.  In case 3, $V$ is selected to minimize $R_1$ along the lines of \cite{Orlitsky01}.

\begin{table}[h]
\caption{Known capacity region (cases where ${\cal R}_{p_0,in} = {\cal R}_{p_0,out}$).}
\label{table relay}
\centering
\begin{tabular}{|l|c|l|}
\hline
& Condition & Auxiliary \\
\hline
Case 1: & $X-Y-Z$ & $U=\emptyset, V=\emptyset$ \\
Case 2: & $Y-X-Z$ & $U=Z, V=\emptyset$ \\
Case 3: & $Z=f(X,Y)$ & $U=\emptyset$ \\
\hline
\end{tabular}
\end{table}

\begin{example}[Task assignment]
\label{example relay task assignment}
Consider again a task assignment setting similar to Example \ref{example cascade}, where three tasks are to be assigned without duplication to the three nodes X, Y, and Z, and the assignments for nodes X and Y are chosen uniformly at random by nature among all pairs of tasks where $X \neq Y$.  A distribution capturing this coordination behavior is the uniform distribution over the six permutations of task assignments.  Let $\hat{p}_0(x,y)$ be the distributions obtained by sampling $X$ and $Y$ uniformly at random from the set $\{1,2,3\}$ without replacement, and let $\hat{p}(z|x,y)$ be the degenerate distribution where $Z$ is the remaining unassigned task in $\{1,2,3\}$.  Figure \ref{figure relay three card} illustrates a valid outcome of the task assignments.
\begin{figure}[h]
\psfrag{l1}[][][0.8]{$\;\;\;\;\; X$}
\psfrag{l2}[][][0.8]{$Y \;\;\;\;\;$}
\psfrag{l3}[][][0.8]{$Z \;\;\;\;\;$}
\psfrag{l4}[][][0.8]{}
\psfrag{l5}[][][0.8]{$R_1$}
\psfrag{l6}[][][0.8]{$R_2$}
\psfrag{c1}[][][0.8]{$3$}
\psfrag{c2}[][][0.8]{$1$}
\psfrag{c3}[][][0.8]{$2$}
\centerline{\includegraphics[width=6cm]{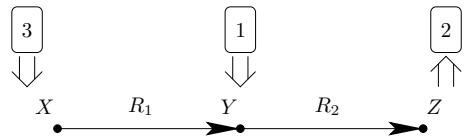}}
\caption{{\em Task assignment in the cascade multiterminal network.}  Three tasks, numbered $1$, $2$, and $3$, are distributed among three nodes X, Y, and Z in the cascade multiterminal network setting.  The task assignments for X and Y are given randomly by nature but different from each other.  What rates $R_1$ and $R_2$ are necessary to allow Z to choose a different task from both X and Y?}
\label{figure relay three card}
\end{figure}

Task assignment in the cascade multiterminal network amounts to computing a function $Z(X,Y)$, and the bounds in Theorem \ref{theorem relay} are tight in such cases.  The rate-coordination region ${\cal R}_{\hat{p}_0}(\hat{p}(z|x,y))$ is given by
\begin{eqnarray*}
{\cal R}_{\hat{p}_0}(\hat{p}(z|x,y)) & = & \left\{ (R_1,R_2) \; : \;
\begin{array}{l}
R_1 \geq \log 2, \\
R_2 \geq \log 3.
\end{array}
\right\}.
\end{eqnarray*}
This is achieved by letting $U=\emptyset$ and $V=X$ in the definition of ${\cal C}_{p_0,in}$.  To show that this region meets the outer bound ${\cal C}_{p_0,out}$, make the observation that $I(X;U|Y) \geq I(X;Z|Y)$ in relation to the bound on $R_1$, since $X - (Y,U) - Z$ forms a Markov chain.
\end{example}

%% file: coordination_strong_coordination.tex
\section{Strong Coordination}
\label{section strong coordination}

So far we have examined coordination where the goal is to generate $Y^n$ through communication based on $X^n$ so that the joint type $P_{X^n,Y^n}(x,y)$ is equal to the desired distribution $p_0(x)p(y|x)$.  This goal relates to the joint behavior at the nodes in the network averaged over time.  There is no imposed requirement that $Y^n$ be random, and the order of the sequence of the $(X_i,Y_i)$ pairs doesn't matter.

How different does the problem become if we actually want the actions at the various nodes in the network to be random according to a desired joint distribution?  In this vein, we turn to a stronger notion of cooperation which we call strong coordination.  We require that the induced distribution over the entire coding block $\tilde{p}(x^n,y^n)$ (induced by the coordination code) be close to the target distribution $p(x^n,y^n) = \prod_{i=1}^n p_0(x_i)p(y_i|x_i)$---so close that a statistician could not tell the difference, based on $(X^n,Y^n)$, of whether $(X^n,Y^n) \sim \tilde{p}(x^n,y^n)$ or $(X^n,Y^n) \sim p(x^n,y^n)$.

Clearly this new strong coordination objective is more demanding than empirical coordination---after all, if one were to generate random actions, i.i.d. in time, according to the appropriate joint distribution, then the empirical distribution would also follow suit.  But in some settings it is crucial for the coordinated behavior to be random.  For example, in situations where an adversary is involved, it might be important to maintain a mystery in the sequence of actions that are generated in the network.

Strong coordination has applications in cooperative game theory, discussed in \cite{cuff08}.  Suppose a team shares the same payoff in a repeated game setting.  An opponent who tries to anticipate and exploit patterns in the team's combined actions will be adequately combatted by strong coordination according to a well-chosen joint distribution.

\subsection{Problem specifics}
\label{subsection strong definitions}

Most of the definitions relating to empirical coordination in Section \ref{subsection definitions} carry over to strong coordination, including the notions of coordination codes and induced distributions.  However, in the context of strong coordination, achievability has nothing to do with the joint type.  Here we define strong achievability to mean that the distribution of the time-sequence of actions in the network is close in total variation to the desired joint distribution, i.i.d. in time.  We discuss the strong coordination capacity region $\underline{\cal C}_{p_0}$, like the region of Definition \ref{definition coordination capacity region}, but instead defined by this notion of strong achievability.

\begin{definition}[Strong achievability]
\label{definition strong achievability}
A desired distribution $p(x,y,z)$ is strongly achievable if there exists a sequence of (non-deterministic) coordination codes such that the total variation between the induced distribution $\tilde{p}(x^n,y^n,z^n)$ and the i.i.d. desired distribution goes to zero.  That is,
\begin{eqnarray*}
\left\| \tilde{p}(x^n,y^n,z^n) - \prod_{i=1}^n p(x_i,y_i,z_i) \right\|_{TV} \longrightarrow 0.
\end{eqnarray*}
\end{definition}

A non-deterministic coordination code is a deterministic code that utilizes an extra argument for each encoder and decoder which is a random variable independent of all the other variables and actions.  It seems quite reasonable to allow the encoders and decoders to use private randomness during the implementation of the coordination code.  This allowance would have also been extended to the empirical coordination framework of sections \ref{section weak coordination}, \ref{section complete results}, and \ref{section partial results}; however, randomized encoding and decoding is not beneficial in that framework because the objective has nothing to do with producing random actions (appropriately distributed).  This claim is similar to Theorem \ref{theorem no common randomness}.  Thus, non-deterministic coordination codes do not improve the empirical coordination capacity over deterministic coordination codes.

Common randomness plays a crucial role in achieving strong coordination.  For instance, in a network with no communication, only independent actions can be generated at each node without common randomness, but actions can be generated according to any desired joint distribution if enough common randomness is available, as is illustrated in Figure \ref{f_all_dist} of Section \ref{section introduction}.  In addition, for each desired joint distribution we can identify a specific bit-rate of common randomness that must be available to the nodes in the network.  This motivates us to deal with common randomness more precisely.

Aside from the communication in the network, we allow common randomness to be supplied to each node.  However, to quantify the amount of common randomness, we limit it to a rate of $R_0$ bits per action.  For an $n$-block coordination code, $\omega$ is uniformly distributed on the set $\Omega = \{1,...,2^{nR_0}\}$.  In this way, common randomness is viewed as a resource alongside communication.

\subsection{Preliminary observations}
\label{subsection strong observations}

The strong coordination capacity region $\underline{\cal C}_{p_0}$ is not convex in general.  This becomes immediately apparent when we consider a network with no communication and without any common randomness.  An arbitrary joint distribution is not strongly achievable without communication or common randomness, but any extreme point in the probability simplex corresponds to a degenerate distribution that is trivially achievable.  Thus we see that convex combinations of achievable points in the strong coordination capacity region are not necessarily strongly achievable, and cannot be achieved through simple time-sharing as was done for empirical coordination.

We use total variation as a measurement of fidelity for the distribution of the actions in the network.  This has a number of implications.  If two distributions have a small total variation between them, then a hypothesis test cannot reliably tell them apart.  Additionally, the expected value of a bounded function of these random variables cannot differ by much.  Steinberg and Verd{\'u}, for example, also use total variation as one of a handful of fidelity criteria when considering the simulation of random variables in \cite{steinberg}.  On the other hand, Wyner used normalized relative entropy as his measurement of error for generating random variables in \cite{wyner}.  Neither quantity, total variation or normalized relative entropy, is dominated by the other in general (because of the normalization).  However, relative entropy would give infinite penalty if the support of the block-distribution of actions is not contained in the support of the desired joint distribution.  We find cases where the rates required under the constraint of normalized relative entropy going to zero are unpleasantly high.  For instance, lossless source coding would truly have to be lossless, with zero error.

Based on the success of random codebooks in information theory and source coding in particular, it seems hopeful that we might always be able to use common randomness to augment a coordination code intended for empirical coordination to result in a randomized coordination code that achieves strong coordination.  Bennett et. al. demonstrate this principle for the two-node setting with their reverse Shannon theorem \cite{bennett2002}.  They use common randomness to generate a random codebook.  Then the encoder synthesizes a memoryless channel and finds a sequence in the codebook with the same joint type as the synthesized output.  Will methods like this work in other network coordination settings as well?  The following conjecture makes this statement precise and is consistent with both networks considered for strong coordination in this section of the paper.
\begin{conjecture}[Strong meets empirical coordination]
\label{conjecture strong equals empirical}
With enough common randomness, for instance if $\omega \sim \mbox{Unif}\{[0,1]\}$, the strong coordination capacity region is the same as the empirical coordination capacity region for any specific network setting .  That is,
\begin{eqnarray*}
\mbox{With unlimited common randomness: } \phantom{space} \underline{\cal C}_{p_0} & = & {\cal C}_{p_0}.
\end{eqnarray*}
\end{conjecture}

If Conjecture \ref{conjecture strong equals empirical} is true, then results regarding empirical coordination should influence strong coordination schemes, and strong coordination capacity regions will reduce to empirical coordination capacity regions under the appropriate limit.

\subsection{No communication}
\label{subsection no communication strong}

Here we characterize the strong coordination capacity region $\underline{\cal C}$ for the no communication network of Figure \ref{figure no communication strong}.  A collection of nodes X, Y, and Z generate actions according to the joint distribution $p(x,y,z)$ using only common randomness (and private randomization).  The strong coordination capacity region characterizes the set of joint distributions that can be achieved with common randomness at a rate of $R_0$ bits per action.

\begin{figure}[h]
\psfrag{X}[][][.8]{$X$}
\psfrag{Y}[][][.8]{$Y$}
\psfrag{Z}[][][.8]{$Z$}
\psfrag{Dist}[][][.8]{$|\Omega| = 2^{nR_0}$}
\centerline{\includegraphics[width=6cm]{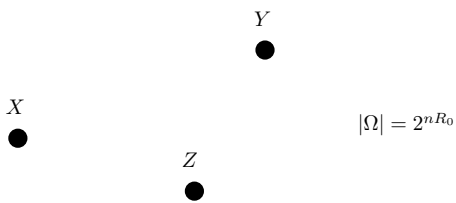}}
\caption{{\em No communication.}  Three nodes generate actions $X$, $Y$, and $Z$ according to $p(x,y,z)$ without communication.  The rate of common randomness needed is characterized in Theorem \ref{theorem no communication strong}.}
\label{figure no communication strong}
\end{figure}

Wyner considered a two-node setting in \cite{wyner}, where correlated random variables are constructed based on common randomness.  He found the amount of common randomness needed and named the quantity ``common information.''  Here we extend that result to three nodes, and the conclusion for any number of nodes is immediately apparent.

The $n$-block coordination codes consist of three non-deterministic decoding functions,
\begin{eqnarray*}
x^n & : & \{1,...,2^{nR_0}\} \longrightarrow {\cal X}^n, \\
y^n & : & \{1,...,2^{nR_0}\} \longrightarrow {\cal Y}^n, \\
z^n & : & \{1,...,2^{nR_0}\} \longrightarrow {\cal Z}^n.
\end{eqnarray*}
Each function can use private randomization to probabilistically map the common random bits $\omega$ to action sequences.  That is, the functions $x^n(\omega)$, $y^n(\omega)$, and $z^n(\omega)$ behave according to conditional probability mass functions $p(x^n|\omega)$, $p(y^n|\omega)$, and $p(z^n|\omega)$.

The rate region given in Theorem \ref{theorem no communication strong} can be generalized to any number of nodes.

\begin{theorem}[Strong coordination capacity region]
\label{theorem no communication strong}
The strong coordination capacity region $\underline{\cal C}$ for the no communication network of Figure \ref{figure no communication strong} is given by
\begin{eqnarray*}
\underline{\cal C} & = & \left\{
\begin{array}{l}
p(x,y,z) \; : \; \exists p(u|x,y,z) \mbox{ such that} \\
p(x,y,z,u) = p(u)p(x|u)p(y|u)p(z|u) \\
|{\cal U}| \leq |{\cal X}||{\cal Y}||{\cal Z}|, \\
R_0 \geq I(X,Y,Z;U).
\end{array}
\right\}.
\end{eqnarray*}
\end{theorem}

{\em Discussion:}
The proof of Theorem \ref{theorem no communication strong}, sketched in Section \ref{section proofs}, follows nearly the same steps as Wyner's common information proof.  This generalization can be interpreted as a proposed measurement of common information between a group of random variables.  Namely, the amount of common randomness needed to generate a collection of random variables at isolated nodes is the amount of common information between them.  However, it would also be interesting to consider a richer problem by allowing each subset of nodes to have an independent common random variable and investigating all of the rates involved.

\begin{example}[Task assignment]
\label{example no communication strong}
Suppose there are tasks numbered $1,...,k$, and three of them are to be assigned randomly to the three nodes X, Y, and Z without duplication.  That is, the desired distribution $\hat{p}(x,y,z)$ for the three actions in the network is the distribution obtained by sampling $X$, $Y$, and $Z$ uniformly at random from the set $\{1,...,k\}$ without replacement.  The three nodes do not communicate but have access to common randomness at a rate of $R_0$ bits per action.  We want to determine the infimum of rates $R_0$ required to strongly achieve $\hat{p}(x,y,z)$.  Figure \ref{figure no communication three card strong} illustrates a valid outcome of the task assignments.

\begin{figure}[h]
\psfrag{l1}[][][0.8]{$X$}
\psfrag{l2}[][][0.8]{$Z$}
\psfrag{l3}[][][0.8]{$Y$}
\psfrag{l7}[][][0.8]{$k \geq 6:$}
\psfrag{c1}[][][0.8]{$3$}
\psfrag{c2}[][][0.8]{$2$}
\psfrag{c3}[][][0.8]{$6$}
\psfrag{l8}[][][0.8]{$|\Omega| = 2^{nR_0}$}
\centerline{\includegraphics[width=.4\textwidth]{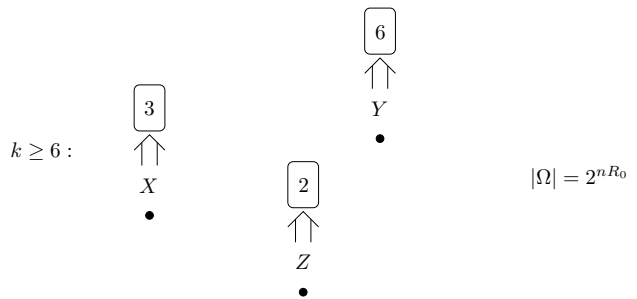}}
\caption{{\em Random task assignment with no communication.}  A task from a set of tasks numbered $1,...,k$ is to be assigned randomly but uniquely to each of the nodes X, Y, and Z without any communication between them.   The rate of common randomness needed to accomplish this is roughly $R_0 \geq 3 \log 3$ for large $k$.}
\label{figure no communication three card strong}
\end{figure}

Theorem \ref{theorem no communication strong} tells us which values of $R_0$ will result in $\hat{p}(x,y,z) \in \underline{\cal C}$.  We must optimize over distributions of an auxiliary random variable $U$.  Two things come in to play to make this optimization manageable:  The variables $X$, $Y$, and $Z$ are all conditionally independent given $U$; and the distribution $\hat{p}$ has sparsity.  For any particular value of $U$, the conditional supports of $X$, $Y$, and $Z$ must be disjoint.  Therefore,
\begin{eqnarray*}
I(X,Y,Z;U) & = & H(X,Y,Z) - H(X,Y,Z|U) \\
& = & H(X,Y,Z) - {\bf E} \left[ H(X,Y,Z|U=u) \right] \\
& \geq & H(X,Y,Z) - {\bf E} \left[ \log (k_{1,U} k_{2,U} k_{3,U}) \right], \\
\end{eqnarray*}
where $k_{1,U}$, $k_{2,U}$, and $k_{3,U}$ are integers that sum to $k$ for all $U$.  Therefore, we maximize $\log(k_{1,U}k_{2,U}k_{3,U})$ by letting the three integers be as close to equal as possible.  Furthermore, it is straightforward to find a joint distribution that meets this inequality with equality.

If $k$, the number of tasks, is divisible by three, then we see that $\hat{p}(x,y,z) \in \underline{\cal C}$ for values of $R_0 > 3 \log 3 - \log(\frac{k}{k-1}) - \log(\frac{k}{k-2})$.  No matter how large $k$ is, the required rate never exceeds $R_0 > 3 \log 3$.
\end{example}

\subsection{Two nodes}
\label{subsection two nodes strong}

We can revisit the two-node network from Section \ref{subsection two nodes} and ask what communication rate is needed for strong coordination.  In this network the action at node X is specified by nature according to $p_0(x)$, and a message is sent from node X to node Y at rate $R$.  Common randomness is also available to both nodes at rate $R_0$.  The common randomness is independent of the action $X$.

\begin{figure}[h]
\psfrag{l1}[][][.8]{$X \sim p_0(x)$}
\psfrag{l2}[][][.8]{$Y$}
\psfrag{l4}[][][.8]{$R$}
\psfrag{l6}[][][.7]{Node X}
\psfrag{l7}[][][.7]{Node Y}
\psfrag{l8}[][][.8]{$|\Omega| = 2^{nR_0}$}
\centerline{\includegraphics[width=.35\textwidth]{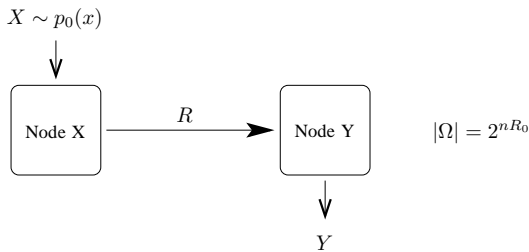}}
\caption{{\em Two nodes.}  The action at node X is specified by nature according to $p_0(x)$, and a message is sent from node X to node Y at rate $R$.  Common randomness is also available to both nodes at rate $R_0$.  The common randomness is independent of the action $X$.  The strong coordination capacity region $\protect\underline{\cal C}$ depends on the amount of common randomness available.  With no common randomness, $\protect\underline{\cal C}$ contains all rate-coordination pairs where the rate is greater than the common information between $X$ and $Y$.  With enough common randomness, $\protect\underline{\cal C}$ contains all rate-coordination pairs where the rate is greater than the mutual information between $X$ and $Y$.}
\label{figure two nodes strong}
\end{figure}

The rates $R_0$ and $R$ required for strong coordination in the two-node network are characterized in \cite{cuff08} and were independently discovered by Bennett et. al. \cite{devetak} in the context of synthesizing a memoryless channel.  Here we take particular note of the two extremes:  what is the strong coordination capacity region when no common randomness is present, and how much common randomness is enough to maximize the strong coordination capacity region?

The $(2^{nR},n)$ coordination codes consist of a non-deterministic encoding function,
\begin{eqnarray*}
i & : & {\cal X}^n \times \{1,...,2^{nR_0}\} \longrightarrow \{1,...,2^{nR}\}.
\end{eqnarray*}
and a non-deterministic decoding function,
\begin{eqnarray*}
y^n & : & \{1,...,2^{nR}\} \times \{1,...,2^{nR_0}\} \longrightarrow {\cal Y}^n.
\end{eqnarray*}
Both functions can use private randomization to probabilistically map the arguments onto the range of the function.  That is, the encoding function $i(x^n,\omega)$ behaves according to a conditional probability mass function $p(i|x^n,\omega)$, and the decoding function $y^n(i,\omega)$ behaves according to a conditional probability mass function $p(y^n|i,\omega)$.

The actions $X^n$ are chosen by nature i.i.d. according to $p_0(x)$, and the actions $Y^n$ are constructed by implementing the non-deterministic coordination code as
\begin{eqnarray*}
Y^n & = & y^n(i(X^n,\omega),\omega).
\end{eqnarray*}

Let us define two quantities before stating the result.  The first is Wyner's common information $C(X;Y)$ \cite{wyner}, which turns out to be the communication rate requirement for strong coordination in the two-node network when no common randomness is available:
\begin{eqnarray*}
C(X;Y) & \triangleq & \min_{U \; : \; X - U - Y} I(X,Y;U),
\end{eqnarray*}
where the notation $X-U-Y$ represents a Markov chain from $X$ to $U$ to $Y$.
The second quantity we call {\em necessary conditional entropy} $H(Y \dag X)$, which we will show to be the amount of common randomness needed to maximize the strong coordination capacity region in the two-node network:
\begin{eqnarray*}
H(Y \dag X) & \triangleq & \min_{f \; : \; X - f(Y) - Y} H(f(Y)|X).
\end{eqnarray*}

\begin{theorem}[Strong coordination capacity region]
\label{theorem two nodes strong}
With no common randomness, $R_0 = 0$, the strong coordination capacity region $\underline{\cal C}_{p_0}$ for the two-node network of Figure \ref{figure two nodes strong} is given by
\begin{eqnarray*}
\underline{\cal C}_{p_0} & = & \left\{ (R, p(y|x)) \; : \; R \geq C(X;Y) \right\}.
\end{eqnarray*}
On the other hand, if and only if the rate of common randomness is greater than the necessary conditional entropy, $R_0 \geq H(Y \dag X)$, the strong coordination capacity region $\underline{\cal C}_{p_0}$ for the two-node network of Figure \ref{figure two nodes strong} is given by
\begin{eqnarray*}
\underline{\cal C}_{p_0} & = & \left\{ (R, p(y|x)) \; : \; R \geq I(X;Y) \right\}.
\end{eqnarray*}
\end{theorem}

{\em Discussion:}
The proof of Theorem \ref{theorem two nodes strong}, found in Section \ref{section proofs}, is an application of Theorem 3.1 in \cite{cuff08}.  This theorem is consistent with Conjecture \ref{conjecture strong equals empirical}---with enough common randomness, the strong coordination capacity region $\underline{\cal C}_{p_0}$ is the same as the coordination capacity region ${\cal C}_{p_0}$ found in Section \ref{subsection two nodes}.

For many joint distributions, the necessary conditional entropy $H(Y \dag X)$ will simply equal the conditional entropy $H(Y|X)$.

\begin{example}[Task assignment]
\label{example two nodes strong}
Consider again a task assignment setting similar to Example \ref{example no communication strong}, where tasks are numbered $1,...,k$ and are to be assigned randomly to the two nodes $X$ and $Y$ without duplication.  The action $X$ is supplied by nature, uniformly at random ($\hat{p}_0(x)$), and the desired distribution $\hat{p}(y|x)$ for the action $Y$ is the uniform distribution over all tasks not equal to $X$.  Figure \ref{figure two nodes three card strong} illustrates a valid outcome of the task assignments.

\begin{figure}[h]
\psfrag{l1}[][][0.8]{$X \;$}
\psfrag{l2}[][][0.8]{$Y$}
\psfrag{l4}[][][0.8]{}
\psfrag{l5}[][][0.8]{$R$}
\psfrag{l7}[][][0.8]{$k \geq 7:$}
\psfrag{c1}[][][0.8]{$2$}
\psfrag{c2}[][][0.8]{$7$}
\psfrag{l8}[][][0.8]{$|\Omega| = 2^{nR_0}$}
\centerline{\includegraphics[width=.4\textwidth]{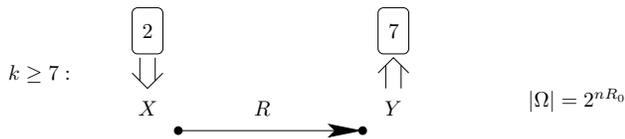}}
\caption{{\em Task assignment in the two-node network.}  A task from a set of tasks numbered $1,...,k$ is to be assigned randomly but uniquely to each of the nodes X and Y in the two-node network.  The task assignment for X is given by nature.  Common randomness at rate $R_0$ is available to both nodes, and a message is sent from node X to node Y at rate $R$.  When no common randomness is available, the required communication rate is $R \geq 2 - \log ( \frac{k}{k-1} )$ bits (for even $k$).  At the other extreme, if the rate of common randomness is greater than $\log(k-1)$, then $R \geq \log ( \frac{k}{k-1} )$ suffices.}
\label{figure two nodes three card strong}
\end{figure}

To apply Theorem \ref{theorem two nodes strong} we must evaluate the three quantities $I(X;Y)$, $C(X;Y)$, and $H(Y \dag X)$.  For the joint distribution $\hat{p}_0(x)\hat{p}(y|x)$, the necessary conditional entropy $H(Y \dag X)$ is exactly the conditional entropy $H(Y|X)$.  The computation of the common information $C(X;Y)$ follows the same steps as the derivation found in Example \ref{example no communication strong}.  Let $\lceil k \rceil$ take the value of $k$ rounded up to the nearest even number.
\begin{eqnarray*}
I(X;Y) & = & \log \left( \frac{k}{k-1} \right), \\
C(X;Y) & = & 2 \mbox{ bits } - \log \left( \frac{\lceil k \rceil}{\lceil k \rceil -1} \right), \\
H(Y \dag X) & = & \log \left( k-1 \right).
\end{eqnarray*}

Without common randomness, we find that the communication rate $R \geq 2 \mbox{ bits } - \log \left( \frac{\lceil k \rceil}{\lceil k \rceil -1} \right)$ is necessary to strongly achieve $\hat{p}_0(x)\hat{p}(y|x)$.  The strong coordination capacity region $\underline{\cal C}_{\hat{p}_0}$ expands as the rate of common randomness $R_0$ increases.  Additional common randomness is no longer useful when $R_0 > \log(k-1)$.  With this amount of common randomness, only the communication rate $R \geq \log (\frac{k}{k-1})$ is necessary to strongly achieve $\hat{p}_0(x)\hat{p}(y|x)$.
\end{example}

%% file: coordination_rate_distortion.tex
\section{Rate-distortion Theory}
\label{section rate distortion}

The challenge of describing random sources of information with the fewest bits possible can be defined in a number of different ways.  Traditionally, source coding in networks follows the path of rate-distortion theory by establishing multiple distortion penalties for the multiple sources and reconstructions in the network.  Yet, fundamentally, the rate-distortion problem is intimately connected to empirical coordination.

The basic result of rate-distortion theory for a single memoryless source states that in order to achieve any desired distortion level you must find an appropriate conditional distribution of the reconstruction $\hat{X}$ given the source $X$ and then use a communication rate larger than the mutual information $I(X;\hat{X})$.  This lends itself to the interpretation that optimal encoding for a rate-distortion setting really comes down to coordinating a reconstruction sequence with a source sequence according to a selected joint distribution.  Here we make that observation formal by showing that in general, even in networks, the rate-distortion region is a projection of the coordination capacity region.

The coordination capacity region ${\cal C}_{p_0}$ is a set of rate-coordination tuples.  We can express rate-coordination tuples as vectors.  For example, in the cascade network of Section \ref{subsection cascade} there are two rates $R_1$ and $R_2$.  The actions in this network are $X$, $Y$, and $Z$, where $X$ is given by nature.  Order the space ${\cal X} \times {\cal Y} \times {\cal Z}$ in a sequence $(x_1,y_1,z_1),...,(x_m,y_m,z_m)$, where $m = |{\cal X}||{\cal Y}||{\cal Z}|$.  The rate-coordination tuples $(R_1,R_2,p(y,z|x))$ can be expressed as vectors $[R_1,R_2,p(y_1,z_1|x_1),...,p(y_m,z_m|x_m)]^T$.

The rate-distortion region ${\cal D}_{p_0}$ is the closure of the set of rate-distortion tuples that are achievable in a network.  We say that a distortion $D$ is achievable if there exists a rate-distortion code that gives an expected average distortion less than $D$, using $d$ as a distortion measurement.  For example, in the cascade network of Section \ref{subsection cascade} we might have two distortion functions:  The function $d_1(x,y)$ measures the distortion in the reconstruction at node Y; the function $d_2(x,y,z)$ evaluates distortion jointly between the reconstructions at nodes Y and Z.  The rate-distortion region ${\cal D}_{p_0}$ would consist of tuples $(R_1,R_2,D_1,D_2)$, which indicate that using rates $R_1$ and $R_2$ in the network, a source distributed according to $p_0(x)$ can be encoded to achieve no more than $D_1$ expected average distortion as measured by $d_1$ and $D_2$ distortion as measured by $d_2$.

The relationship between the rate-distortion region ${\cal D}_{p_0}$ and the coordination capacity region ${\cal C}_{p_0}$ is that of a linear projection.  Suppose we have multiple finite-valued distortion functions $d_1,...,d_k$.  We construct a distortion matrix $D$ using the same enumeration $(x_1,y_1,z_1),...,(x_m,y_m,z_m)$ of the space ${\cal X} \times {\cal Y} \times {\cal Z}$ as was used to vectorize the tuples in ${\cal C}_{p_0}$:
\begin{eqnarray*}
D & \triangleq & \phantom{fill space with a lot of nothing much, so I}
\end{eqnarray*}
\begin{eqnarray*}
& \left[
\begin{array}{ccc}
d_1(x_1,y_1,z_1)p_0(x_1) & \cdots & d_1(x_m,y_m,z_m)p_0(x_m) \\
\vdots & \vdots & \vdots \\
d_k(x_1,y_1,z_1)p_0(x_1) & \cdots & d_k(x_m,y_m,z_m)p_0(x_m)
\end{array}
\right]. &
\end{eqnarray*}
The distortion matrix $D$ is embedded in a block diagonal matrix $A$ where the upper-left block is the identity matrix $I$ with the same dimension as the number of rates in the network:
\begin{eqnarray*}
A & \triangleq & \left[
\begin{array}{cc}
I & 0 \\
0 & D
\end{array}
\right].
\end{eqnarray*}

\begin{theorem}[Rate-distortion region]
\label{theorem rate distortion}
The rate-distortion region ${\cal D}_{p_0}$ for a memoryless source with distribution $p_0$ in any rate-limited network is a linear projection of the coordination capacity region ${\cal C}_{p_0}$ by the matrix $A$,
\begin{eqnarray*}
{\cal D}_{p_0} & = & A \; {\cal C}_{p_0}.
\end{eqnarray*}
We treat the elements of ${\cal D}_{p_0}$ and ${\cal C}_{p_0}$ as vectors, as discussed, and the matrix multiplication by $A$ is the standard set multiplication.
\end{theorem}

{\em Discussion:}
The proof of Theorem \ref{theorem rate distortion} can be found in Section \ref{section proofs}.  Since the coordination capacity region ${\cal C}_{p_0}$ is a convex set, the rate-distortion region ${\cal D}_{p_0}$ is also a convex set.

Clearly we can use a coordination code to achieve the corresponding distortion in a rate-distortion setting.  But the theorem makes a stronger statement.  It says that there is not a more efficient way of satisfying distortion limits in any network setting with memoryless sources than by using a code that produces the same joint type for almost every observation of the sources.  It is conceivable that a rate-distortion code for a network setting would produce a variety of different joint types, each satisfying the distortion limit, but varying depending on the particular source sequence observed.  However, given such a rate-distortion code, repeated uses will produce a longer coordination code that consistently achieves coordination according to the expected joint type.  The expected joint type of a good rate-distortion code can be shown to satisfy the distortion constraints.

\begin{figure}[h]
\psfrag{p0}[][][0.8]{$P(Y=1|X=0)$}
\psfrag{p1}[][][0.8]{$P(Y=1|X=1)$}
\psfrag{p00}[][][0.8]{$p_0$} \psfrag{p11}[][][0.6]{$p_1$}
\psfrag{R=0.1}[][][0.8]{$R=0.1$} \psfrag{R=0.4}[][][0.6]{$0.4$}
\psfrag{R=0.7}[][][0.8]{$0.7$}
\psfrag{R(pYX)}[][][0.8]{$R(p_0,p_1)$}
\psfrag{EE}[][][0.8]{$d \leq D$}
\centerline{\includegraphics[width=.4\textwidth]{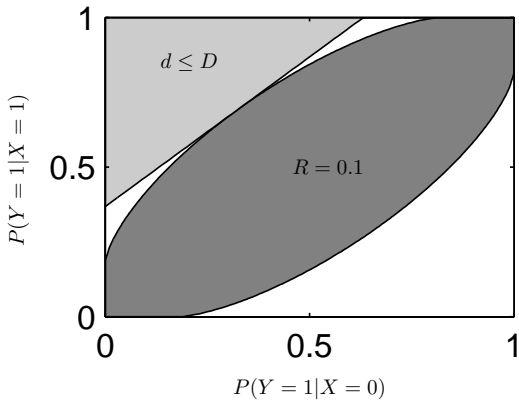}}
\caption{{\em Coordination capacity and rate-distortion.}  The coordination-rate region for a uniform binary source $X$ and binary action $Y$, where $X$ is described at rate $R = 0.1$ bits to node Y in the two-node network.  The shaded region shows distributions with Hamming distortion less than $D$, where $D$ is chosen to satisfy $R(D) = 0.1$ bits.}
\label{figure rate distortion}
\end{figure}

Geometrically, each distortion constraint defines a hyperplane that divides the coordination-rate region into two sets---one that satisfies the distortion constraint and one that does not.  Therefore, minimizing the distortion for fixed rates in the network amounts to finding optimal extreme points in the coordination-rate region in the directions orthogonal to these hyperplanes.  Figure \ref{figure rate distortion} shows the coordination-rate region for $R = 0.1$ bits in the two-node network of Section \ref{subsection two nodes}, with a uniform binary source $X$ and binary $Y$.  The figure also shows the region satisfying a Hamming distortion constraint $D$.

%% file: coordination_appendix.tex
\section{Proofs}
\label{section proofs}

\subsection{Empirical Coordination - Achievability (Sections \ref{section complete results}, \ref{section partial results})}

For a distribution $p(x)$, define the {\em typical set} \TS{} with respect to $p(x)$ to be sequences $x^n$ whose types are $\epsilon$-close to $p(x)$ in total variation.  That is,
\begin{eqnarray}
\label{equation typical set}
\TS{} & \triangleq & \{ x^n \in {\cal X}^n \; : \; \| P_{x^n}(x) - p(x) \|_{TV} < \epsilon \}.
\end{eqnarray}
This definition is almost the same as the definition of the strongly typical set ${\cal A}_{\epsilon}^{*(n)}$ found in (10.106) of Cover and Thomas \cite{CovThom06}, and it shares the same important properties.  The difference is that here we give a total variation constraint ($L_1$ distance) on the type of the sequence rather than an element-wise constraint ($L_{\infty}$ distance).\footnote{Additionally, our definition of the typical set handles the zero probability events more liberally, but this doesn't present any serious complications.}  We deal with \TS{} since it relates more closely to the definition of achievability in Definition \ref{definition achievability}.  However, the sets are almost the same, as the following sandwich suggests:
\begin{eqnarray*}
\nonumber
{\cal A}_{\epsilon}^{*(n)} \; \subset \; \TS{} \; \subset \; {\cal A}_{\epsilon |{\cal X}|}^{*(n)}.
\end{eqnarray*}

A jointly typical set with respect to a joint distribution $p(x,y)$ inherits the same definition as (\ref{equation typical set}), where total variation of the type is measured with respect to the joint distribution.  Thus, achieving empirical coordination with respect to a joint distribution is a matter of constructing actions that are {\em $\epsilon$-jointly typical} (i.e. in the jointly typical set \TS{}) with high probability for arbitrary $\epsilon$.

\subsubsection{Strong Markov Lemma}

If $X-Y-Z$ form a Markov chain, and the pair of sequences $x^n$ and $y^n$ are jointly typical as well as the pair of sequences $y^n$ and $z^n$, it is not true in general that the three sequences $x^n$, $y^n$, and $z^n$ are jointly typical as a triple.  For instance, consider any triple $(x^n,y^n,z^n)$ that is jointly typical with respect to a non-Markov joint distribution having marginal distributions $p(x,y)$ and $p(y,z)$.  However, the Markov Lemma \cite{Berger78} states that if $Z^n$ is randomly distributed according to $\prod_{i=1}^n p(z_i|y_i)$, then with high probability it will be jointly typical with both $x^n$ and $y^n$.  This lemma is used to establish joint typicality in source coding settings where side information is not known to the encoder.  Yet, for a network and encoding scheme that is more intricate, the standard Markov Lemma lacks the necessary strength.  Here we introduce a generalization that will help us analyze the layers of ``piggy-back''-style codes \cite{Wyner75_WAK} used in our achievability proofs.\footnote{Through conversation we discovered that similar effort is being made by Young-Han Kim and Abbas El Gamal and may soon be found in the Stanford EE478 Lecture Notes.}

\begin{theorem}[Strong Markov Lemma]
\label{theorem strong markov}
Given a joint distribution $p(x,y,z)$ on the finite alphabet ${\cal X} \times {\cal Y} \times {\cal Z}$ that yields a Markov chain $X-Y-Z$ (i.e. $p(x,y,z) = p(y)p(x|y)p(z|y)$), let $x^n$ and $y^n$ be arbitrary sequences that are $\epsilon$-jointly typical.  Suppose that $Z^n$ is randomly chosen from the set of $z^n$ sequences that are $\epsilon$-jointly typical with $y^n$ and additionally that the distribution of $Z^n$ is permutation-invariant with respect to $y^n$, which is to say, any two sequences $z^n$ and $\tilde{z}^n$ of the same joint type with $y^n$ have the same probability.  That is,
\begin{eqnarray}
\label{definition permutation invariant}
P_{y^n,z^n} = P_{y^n,\tilde{z}^n} & \Rightarrow & P(Z^n = z^n) = P(Z^n = \tilde{z}^n). \;\;\;\;
\end{eqnarray}
Then,
\begin{eqnarray*}
\Pr \left( (x^n,y^n,Z^n) \in {\cal T}_{4\epsilon}^{(n)} \right) & > & \xi_n,
\end{eqnarray*}
where $\xi_n \to 1$ exponentially fast as $n$ goes to infinity.
\end{theorem}

Notice that permutation invariance is a condition satisfied by most random codebook based proof techniques---for instance, encoding schemes based on i.i.d. codebooks tend to be permutation invariant.  To recover the familiar Markov Lemma, let $Z^n$ have a distribution based on $y^n$ according to $\prod_{i=1}^n p(z_i|y_i)$, where $y^n$ is an $\epsilon$-typical sequence.  Due to the A.E.P., $y^n$ and $Z^n$ will be $2\epsilon$-jointly typical with high probability.  Furthermore, Theorem \ref{theorem strong markov} can be invoked because the distribution is permutation invariant.

The key to proving Theorem \ref{theorem strong markov} is found in Lemma \ref{lemma close to markov}, which uses permutation invariance and counting arguments to show that most realizations look empirically Markov.

\begin{lemma}[Markov Tendency]
\label{lemma close to markov}
Let $x^n \in {\cal X}^n$ and $y^n \in {\cal Y}^n$ be arbitrary sequences.  Suppose that the random sequence $Z^n \in {\cal Z}^n$ has a distribution that is permutation-invariant with respect to $y^n$, as in (\ref{definition permutation invariant}).  Then with high probability which only depends on the sizes of the alphabets ${\cal X}$, ${\cal Y}$, and ${\cal Z}$, the joint type $P_{x^n,y^n,Z^n}$ will be $\epsilon$-close to the Markov joint type $P_{x^n,y^n}P_{Z^n|y^n}$.  That is, for any $\epsilon > 0$,
\begin{eqnarray}
\label{equation close to markov}
\left\| P_{x^n,y^n,Z^n} - P_{x^n,y^n}P_{Z^n|y^n} \right\|_{TV} & < & \epsilon,
\end{eqnarray}
with a probability of at least $1 - 2^{- \alpha n + \beta \log n},$ where $\alpha$ and $\beta$ only depend on the alphabet sizes and $\epsilon$.
\end{lemma}

\begin{proof}[Proof of Theorem \ref{theorem strong markov}]
The proof of Theorem \ref{theorem strong markov} relies mainly on Lemma \ref{lemma close to markov} and repeated use of the triangle inequality.  From Lemma \ref{lemma close to markov} we know that with probability approaching one as $n$ tends to infinity, inequality (\ref{equation close to markov}) is satisfied, namely,
\begin{eqnarray*}
\left\| P_{x^n,y^n,Z^n} - P_{x^n,y^n}P_{Z^n|y^n} \right\|_{TV} & < & \epsilon.
\end{eqnarray*}
In this event, we now show that
\begin{eqnarray*}
(x^n,y^n,Z^n) & \in & {\cal T}_{4\epsilon}^{(n)}.
\end{eqnarray*}
By the definition of total variation one can easily show that
\begin{eqnarray*}
\| P_{x^n,y^n}P_{Z^n|y^n} & - & p_{X,Y} P_{Z^n|y^n} \|_{TV} \\
& = & \| P_{x^n,y^n} - p_{X,Y} \|_{TV} \\
& < & \epsilon.
\end{eqnarray*}
Similarly,
\begin{eqnarray*}
\| p_{Y} p_{X|Y} P_{Z^n|y^n} & - & P_{y^n} p_{X|Y} P_{Z^n|y^n} \|_{TV} \\
& = & \| p_{Y} - P_{y^n} \|_{TV} \\
& < & \epsilon.
\end{eqnarray*}
And finally,
\begin{eqnarray*}
\| P_{y^n,Z^n} p_{X|Y} & - & p_{X,Y,Z} \|_{TV} \\
& = & \| P_{y^n,Z^n} - p_{Y,Z} \|_{TV} \\
& < & \epsilon.
\end{eqnarray*}
Thus, the triangle inequality gives
\begin{eqnarray*}
\| P_{x^n,y^n,Z^n} - p_{X,Y,Z} \|_{TV} & < & 4 \epsilon. \qedhere
\end{eqnarray*}
\end{proof}

\begin{proof}[Proof of Lemma \ref{lemma close to markov}]
We start by defining two constants that simplify this discussion.  The first constant, $\alpha$, is the key to obtaining the uniform bound that Lemma \ref{lemma close to markov} provides.
\begin{eqnarray*}
\alpha & \triangleq & \min_{
p(x,y,z) \in {\cal S}_{{\cal X},{\cal Y},{\cal Z}} \; : \; \|p(x,y,z) - p(x,y)p(z|y) \|_{TV} \geq \epsilon
}
I(X;Z|Y), \\
\beta & \triangleq & 2 |{\cal X}| |{\cal Y}| |{\cal Z}|.
\end{eqnarray*}
Here ${\cal S}_{{\cal X},{\cal Y},{\cal Z}}$ is the simplex with dimension corresponding to the product of the alphabet sizes.  Notice that $\alpha$ is defined as a minimization of a continuous function over a compact set; therefore, by analysis we know that the minimum is achieved in the set.  Since $I(X;Z|Y)$ is positive for any distribution that does not form a Markov chain $X-Y-Z$, we find that $\alpha$ is positive for $\epsilon > 0$.  The constants $\alpha$ and $\beta$ are functions of $\epsilon$ and the alphabet sizes $|{\cal X}|$, $|{\cal Y}|$, and $|{\cal Z}|$.

We categorize sequences into sets with the same joint type. The {\em type class} $T_{p(y,z)}$ is defined as
\begin{eqnarray*}
T_{p(y,z)} & \triangleq & \{ (y^n,z^n) \; : \; P_{y^n,z^n} = p(y,z) \}.
\end{eqnarray*}
We also define a {\em conditional type class} $T_{p(z|y)}(y^n)$ to be the set of $z^n$ sequences such that the pair $(y^n,z^n)$ are in the type class $T_{p(y,z)}$.  Namely,
\begin{eqnarray*}
T_{p(z|y)}(y^n) & \triangleq & \{ z^n \; : \; P_{y^n,z^n} = p(z|y) P_{y^n} \}.
\end{eqnarray*}

We will show that the statement made in (\ref{equation close to markov}) is true conditionally for each conditional type class $T_{p(z|y)}(y^n)$ and therefore must be true overall.

Suppose $Z^n$ falls in the conditional type class $T_{P_{\bar{z}^n|y^n}}(y^n)$.  By assumption (\ref{definition permutation invariant}), all $z^n$ in this type class are equally likely.  Assessing probabilities simply becomes a matter of counting.  From the method of types \cite{CovThom06} we know that
\begin{eqnarray*}
\left| T_{P_{\bar{z}^n|y^n}}(y^n) \right| & \geq & n^{- |{\cal Y}| |{\cal Z}| } 2^{n H_{P_{y^n,\bar{z}^n}} (Z|Y)}.
\end{eqnarray*}

We also can bound the number of $z^n$ sequences in $T_{P_{\bar{z}^n|y^n}}(y^n)$ that do not satisfy (\ref{equation close to markov}).  These sequences must fall in a conditional type class $T_{P_{\bar{z}^n|x^n,y^n}}(x^n,y^n)$ where
\begin{eqnarray*}
\left\| P_{x^n,y^n,\bar{z}^n} - P_{x^n,y^n}P_{\bar{z}^n|y^n} \right\|_{TV} \geq \epsilon.
\end{eqnarray*}
For each such type class, the size can be bounded by
\begin{eqnarray*}
\left| T_{P_{\bar{z}^n|x^n,y^n}}(x^n,y^n) \right| & \leq & 2^{n H_{P_{x^n,y^n,\bar{z}^n}} (Z|X,Y)} \\
& = & 2^{n \left( H_{P_{y^n,\bar{z}^n}} (Z|Y) - I_{P_{x^n,y^n,\bar{z}^n}} (X;Z|Y) \right)} \\
& \leq & 2^{n \left( H_{P_{y^n,\bar{z}^n}} (Z|Y) - \alpha \right)}.
\end{eqnarray*}
Furthermore, there are only polynomially many types, bounded by $n^{|{\cal X}| |{\cal Y}| |{\cal Z}| }$.  Therefore, the probability that $Z^n$ does not satisfy (\ref{equation close to markov}) for any conditional type $P_{\bar{z}^n|y^n}$ is bounded by
\begin{eqnarray*}
\Pr ( \; \mbox{not } (\ref{equation close to markov}) \! \! \! \! & | & \! \! \! \! Z^n \in T_{P_{\bar{z}^n|y^n}}(y^n) \; ) \\
& = & \frac{\left| \{ z^n \in T_{P_{\bar{z}^n|y^n}}(y^n) : \mbox{ not } (\ref{equation close to markov}) \} \right|}{\left| T_{P_{\bar{z}^n|y^n}}(y^n) \right|} \\
& \leq & \frac{n^{|{\cal X}| |{\cal Y}| |{\cal Z}|} 2^{n \left( H_{P_{y^n,\bar{z}^n}} (Z|Y) - \alpha \right)}}{n^{- |{\cal Y}| |{\cal Z}| } 2^{n H_{P_{y^n,\bar{z}^n}} (Z|Y)}} \\
& = & n^{|{\cal Y}| |{\cal Z}| + |{\cal X}| |{\cal Y}| |{\cal Z}|} 2^{- \alpha n} \\
& \leq & 2^{- \alpha n + \beta \log n}. \qedhere
\end{eqnarray*}
\end{proof}

\subsubsection{Generic Achievability Proof}

The coding techniques for achieving the empirical coordination regions in Sections \ref{section complete results} and \ref{section partial results} are familiar from rate distortion theory.  For the proofs, we construct random codebooks for communication and show that the resulting encoding schemes perform well on average, producing jointly-typical actions with high probability.  This proves that there must be at least one deterministic scheme that performs well.  Here we prove one generally useful example to verify that the rate-distortion techniques actually do work for achieving empirical coordination.  The technique here is very similar to the source coding technique of ``piggy-back'' codes introduced by Wyner \cite{Wyner75_WAK}.

Consider the two-node source coding setting of Figure \ref{figure two nodes side info} with arbitrary sequences $x^n$, $y^n$, and $z^n$ that are $\epsilon$-jointly typical according to a joint distribution $p(x,y,z)$.  The sequences $x^n$ and $y^n$ are available to the encoder at node 1, while $y^n$ and $z^n$ are available to the decoder at node 2.  We can think of $x^n$ as the source to be encoded and $y^n$ and $z^n$ as side information known to either both nodes or the decoder only, respectively.  Communication from node 1 to node 2 at rate $R$ is used to produce a sequence $U^n$. Original results related to this setting in the context of rate-distortion theory can be found in the work of Wyner and Ziv \cite{WynerZiv76}.  Here we analyze a randomized coding scheme that attempts to produce a sequence $U^n$ at the decoder such that $(x^n,y^n,z^n,U^n)$ are $(8\epsilon)$-jointly typical with respect to a joint distribution of the form $p(x,y,z)p(u|x,y)$.  We give a scheme that uses a communication rate of $R > I(X;U|Y,Z)$ and is successful with probability approaching one as $n$ tends to infinity for all jointly typical sequences $x^n$, $y^n$, and $z^n$.

\begin{figure}[h]
\psfrag{l1}[][][.8]{$x^n,y^n$}
\psfrag{l2}[][][.8]{$U^n$}
\psfrag{l4}[][][.8]{$I \in \left[ 2^{nR} \right]$}
\psfrag{l6}[][][.7]{Node 1}
\psfrag{l7}[][][.7]{Node 2}
\psfrag{l8}[][][.8]{$y^n,z^n$}
\centerline{\includegraphics[width=.25\textwidth]{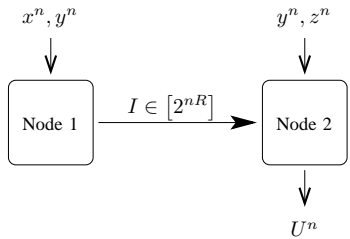}}
\caption{{\em Two nodes with side information.}  This network represents a generic source coding setting encountered in networks and will illustrate standard encoding techniques. The sequences $x^n$, $y^n$, and $z^n$ are jointly typical with respect to $p_0(x,y,z)$.  Only $x^n$ and $y^n$ are observed by the encoder at node 1.  A message is sent to specify $U^n$ to node 2 at rate $R$.  A randomized coding scheme can produce $U^n$ to be jointly typical with $(x^n,y^n,z^n)$ with respect to a Markov chain $Z - (X,Y) - U$ with high probability, regardless of the particular sequences $x^n$, $y^n$, and $z^n$, as long as the rate is greater than the conditional mutual information $I(X;U|Y,Z)$.}
\label{figure two nodes side info}
\end{figure}

The $(2^{nR},n)$ coordination codes consist of a randomized encoding function
\begin{eqnarray*}
i & : & {\cal X}^n \times {\cal Y}^n \times \Omega \longrightarrow \{1,...,2^{nR}\},
\end{eqnarray*}
and a randomized decoding function
\begin{eqnarray*}
u^n & : & \{1,...,2^{nR}\} \times {\cal Y}^n \times {\cal Z}^n \times \Omega \longrightarrow {\cal U}^n.
\end{eqnarray*}
These functions are random simply because the common randomness $\omega$ is involved for generating random codebooks.

The sequences $x^n, y^n$, and $z^n$ are arbitrary jointly typical sequences according to $p_0(x,y,z)$, and the sequence $U^n$ is a randomized function of $x^n, y^n$, and $z^n$ given by implementing the coordination code as
\begin{eqnarray*}
U^n & = & u^n(i(x^n,y^n,\omega), y^n, z^n, \omega).
\end{eqnarray*}

\begin{lemma}[Generic Coordination with Side Information]
\label{lemma generic example}
For the two-node network with side information of Figure \ref{figure two nodes side info} and any discrete joint distribution of the form $p(x,y,z)p(u|x,y)$, there exists a function $\delta(\epsilon)$ which goes to zero as $\epsilon$ goes to zero such that, for any $\epsilon>0$ and rate $R > I(X;U|Y,Z) + \delta(\epsilon)$, there exists a sequence of randomized coordination codes at rate $R$ for which
\begin{eqnarray*}
\Pr \left( (x^n,y^n,z^n,U^n) \in {\cal T}_{\delta(\epsilon)}^{(n)} \right) & \to & 1
\end{eqnarray*}
as $n$ goes to infinity, uniformly for all $(x^n,y^n,z^n) \in {\cal T}_{\epsilon}^{(n)}$.
\end{lemma}

\begin{proof}
Consider a joint distribution $p(x,y,z)p(u|x,y)$ and define $\gamma$ to be the excess rate, $\gamma = R - I(X;U|Y,Z)$.  The conditions of Lemma \ref{lemma generic example} require that $\gamma > \delta(\epsilon)$ for some $\delta(\epsilon)$ that goes to zero as $\epsilon$ goes to zero.  We will identify a valid function $\delta(\epsilon)$ at the conclusion of the following analysis.

We first over-cover the typical set of $(x^n,y^n)$ using a codebook of size $2^{nR_c}$, where $R_c = I(X,Y;U) + \gamma/2$.  We then randomly categorize the codebook sequences into $2^{nR}$ bins, yielding roughly $2^{nR_b}$ sequences in each bin, where
\begin{eqnarray*}
R_b & = & R_c - R \\
& = & I(X,Y;U) - I(X;U|Y,Z) - \gamma/2 \\
& = & I(X,Y,Z;U) - I(X;U|Y,Z) - \gamma/2 \\
& = & I(Y,Z;U) - \gamma/2.
\end{eqnarray*}

{\em Codebook:}
Using $\omega$, generate a codebook ${\mathbb C}$ of $2^{n R_c}$ sequences $u^n(j)$ independently according to the marginal distribution $p(u)$, namely $\prod_{i=1}^n p(u_i)$.  Randomly and independently assign each one a bin number $b(u^n(j))$ in the set $\{1,...,2^{nR}\}$.

{\em Encoder:}
The encoding function $i(x^n,y^n,\omega)$ can be explained as follows.  Search the codebook ${\mathbb C}$ and identify an index $j$ such that $(x^n,y^n,u^n(j)) \in {\cal T}_{2\epsilon}^{(n)}$.  If multiple exist, select the first such $j$.  If none exist, select $j=1$.  Send the bin number $i(x^n,y^n,\omega) = b(u^n(j))$.

{\em Decoder:}
The decoding function $u^n(i,y^n,z^n,\omega)$ can be explained as follows.  Consider the codebook ${\mathbb C}$ and identify an index $j$ such that $(y^n,z^n,u^n(j)) \in {\cal T}_{8\epsilon}^{(n)}$ and $b(u^n(j)) = i$.  If multiple exist, select the first such $j$.  If none exist, select $j=1$.  Produce the sequence $U^n = u^n(j)$.

{\em Error Analysis:}
We conservatively declare errors for any of the following, $E_1$, $E_2$, or $E_3$.

{\em Error 1:  The encoder does not find a $(2\epsilon)$-jointly typical sequence in the codebook.}  By the method of types one can show, as in Lemma 10.6.2 of \cite{CovThom06}, that each sequence in ${\mathbb C}$ is $(2\epsilon)$-jointly typical with $(x^n,y^n)$ with probability greater than $2^{-n(I(X,Y;U) + \delta_1(\epsilon))}$ for $n$ large enough, where $\delta_1(\epsilon)$ goes to zero as $\epsilon$ goes to zero.

Each sequence in the codebook ${\mathbb C}$ is generated independently, so the probability that none of them are jointly typical is bounded by
\begin{eqnarray*}
\Pr(E_1) & \leq & (1 - 2^{-n(I(X,Y;U) + \delta_1(\epsilon))})^{2^{nR_c}} \\
& \leq & e^{-2^{nR_c}2^{-n(I(X,Y;U) + \delta_1(\epsilon))}} \\
& = & e^{-2^{n(R_c - I(X,Y;U) - \delta_1(\epsilon))}} \\
& = & e^{-2^{n(\gamma/2 - \delta_1(\epsilon))}}.
\end{eqnarray*}

{\em Error 2:  The sequence identified by the encoder is not $(8\epsilon)$-jointly typical with $(x^n,y^n,z^n)$.}  Assuming $E_1$ did not occur, because of the Markovity $Z-(X,Y)-U$ implied by $p(x,y,z)p(u|x,y)$ and the symmetry of our codebook construction, we can invoke Theorem \ref{theorem strong markov} to verify that the conditional probability $Pr(E_2 | E_1^c)$ is arbitrarily small for large enough $n$.

{\em Error 3:  The decoder finds more than one eligible action sequence.}  Assume that $E_1$ and $E_2$ did not occur.  If the decoder considers the same index $j$ as the encoder selected, then certainly $u^n(j)$ will be be eligible, which is to say it will be $(8\epsilon)$-jointly typical with $(y^n,z^n)$, and the bin index will match the received message.  For all other sequences in the codebook ${\mathbb C}$, an appeal to the property of iterated expectation indicates that the probability of eligibility is slightly less than the a priori probability that a randomly generated sequence and bin number will yield eligibility (had you not known that it was not the sequence selected by the encoder), which is upper bounded by $2^{-nR}2^{-n(I(Y,Z;U) - \delta_2(\epsilon))}$.  Therefore, by the method of types and the union bound,
\begin{eqnarray*}
\Pr (E_3 | E_1^c, E_2^c) & \leq & 2^{nR_c}2^{-nR}2^{-n(I(Y,Z;U) - \delta_2(\epsilon))} \\
& = & 2^{-n(R - R_c + I(Y,Z;U) - \delta_2(\epsilon))} \\
& = & 2^{-n(I(Y,Z;U) - R_b - \delta_2(\epsilon))} \\
& = & 2^{-n(\gamma/2 - \delta_2(\epsilon))}.
\end{eqnarray*}

Thus we can select $\delta(\epsilon) = \max\{ 2 \delta_1(\epsilon), 2 \delta_2(\epsilon), 8 \epsilon \}$ to make all error terms go to zero and satisfy the lemma.
\end{proof}

With the result of Lemma \ref{lemma generic example} in mind, we can confidently talk about using communication to establish coordination of sequences across links in a network.  Throughout the following explanations we will no longer pay particular attention to the $\epsilon$ in the $\epsilon$-jointly typical set.  Instead, we will simply make reference to the generic jointly typical set, with the assumption that $\epsilon$ is sufficiently small and $n$ is sufficiently large.

\subsubsection{Two nodes - Theorem \ref{theorem two nodes}}

It is clear from Lemma \ref{lemma generic example} that an action sequence $Y^n$ jointly typical with $X^n$ can be specified with high probability using any rate $R > I(X;Y)$.  With high probability $X^n$ will be a typical sequence.  Apply Lemma \ref{lemma generic example} with $Y = Z = \emptyset$.

\subsubsection{Isolated node - Theorem \ref{theorem isolated}}

No proof is necessary, as this is a special case of the cascade network with $R_2 = 0$.

\subsubsection{Cascade - Theorem \ref{theorem cascade}}

The cascade network of Figure \ref{figure cascade} has a sequence $X^n$ given by nature.  The actions $X^n$ will be typical with high probability.  Consider the desired coordination $p(y,z|x)$.  A sequence $Z^n$ can be specified with rate $R_Z > I(X;Z)$ to be jointly typical with $X^n$.  This communication is sent to node Y and forwarded on to node Z.  Additionally, now that every node knows $Z^n$, a sequence $Y^n$ can be specified with rate $R_Y > I(X;Y|Z)$ and sent to node Y.  The rates used are
$R_1 = R_Y + R_Z > I(X;Y,Z)$ and $R_2 = R_Z > I(X;Z)$.
\begin{eqnarray*}
R_1 \; = \; R_Y + R_Z & > & I(X;Y,Z), \\
R_2 \; = \; R_Z & > & I(X;Z).
\end{eqnarray*}

\subsubsection{Degraded source - Theorem \ref{theorem degraded source}}

The degraded source network of Figure \ref{figure degraded source} has a sequence $X^n$ given by nature, known to node X, and another sequence $Y^n$, which is a letter-by-letter function of $X^n$, known to node Y.  Incidentally, $Y^n$ is also known to node X because it is a function of the available information.  The actions $X^n$ and $Y^n$ will be jointly typical with high probability.

Consider the desired coordination $p(z|x,y)$ and choose a distribution for the auxiliary random variable $p(u|x,y,z)$ to help achieve it.  The encoder first specifies a sequence $U^n$ that is jointly typical with $X^n$ and $Y^n$.  This requires a rate $R_U > I(X,Y;U) = I(X;U)$, but with binning we only need a rate of $R_1 > I(X;U|Y)$ to specify $U^n$ from node X to node Y.  Binning is not used when $U^n$ is forwarded to node Z.  Finally, after everyone knows $U^n$, the action sequence $Z^n$ jointly typical with $X^n$, $Y^n$, and $U^n$ is specified to node Z at a rate of $R_2 > I(X,Y;Z|U) = I(X;Z|U)$.  Thus, all rates are achievable which satisfy
\begin{eqnarray*}
R_1 & > & I(X;U|Y), \\
R_2 & > & I(X;Z|U), \\
R_3 \; = \; R_U & > & I(X;U).
\end{eqnarray*}

\subsubsection{Broadcast - Theorem \ref{theorem broadcast}}

The broadcast network of Figure \ref{figure broadcast} has a sequence $X^n$ given by nature, known to node X.  The action sequence $X^n$ will be typical with high probability.

Consider the desired coordination $p(y,z|x)$ and choose a distribution for the auxiliary random variable $p(u|x,y,z)$ to help achieve it.  We will focus on achieving one corner point of the pentagonal rate region.  The encoder first specifies a sequence $U^n$ that is jointly typical with $X^n$ using a rate $R_U > I(X;U)$.  This sequence is sent to both node Y and node Z.  After everyone knows $U^n$, the encoder specifies an action sequence $Y^n$ that is jointly typical with $X^n$ and $U^n$ using rate $R_Y > I(X;Y|U)$.  Finally, the encoder at node X, knowing both $X^n$ and $Y^n$, can specify an action sequence $Z^n$ that is jointly typical with $(X^n,Y^n,U^n)$ using a rate $R_Z > I(X,Y;Z|U)$.  This results in rates
\begin{eqnarray*}
R_1 \; = \; R_U + R_Y & > & I(X;U) + I(X;Y|U) \; = \; I(X;U,Y), \\
R_2 \; = \; R_U + R_X & > & I(X;U) + I(X,Y;Z|U).
\end{eqnarray*}

\subsubsection{Cascade multiterminal - Theorem \ref{theorem relay}}

The cascade multiterminal network of Figure \ref{figure relay} has a sequence $X^n$ given by nature, known to node X, and another sequence $Y^n$ given by nature, known to node Y.  The actions $X^n$ and $Y^n$ will be jointly typical with high probability.

Consider the desired coordination $p(z|x,y)$ and choose a distribution for the auxiliary random variables $U$ and $V$ according to the inner bound in Theorem \ref{theorem relay}.  That is, $p(x,y,z,u,v) = p(x,y)p(u,v|x)p(z|y,u,v)$.  We specify a sequence $U^n$ to be jointly typical with $X^n$.  By the Strong Markov Lemma (Theorem \ref{theorem strong markov}), in conjunction with the symmetry of our random coding scheme and the Markovity of the distribution $p(x,y)p(u|x)$, the sequence $U^n$ will be jointly typical with the pair $(X^n,Y^n)$ with high probability.  Using binning, we only need a rate of $R_{U,1} > I(X;U|Y)$ to specify $U^n$ from node X to node Y (as in Lemma \ref{lemma generic example}).  However, we cannot use binning for the message to node Z, so we send the index of the codework itself at a rate of $R_{U,2} > I(X;U)$.  Now that everyone knows the sequence $U^n$, it is treated as side information.

A second auxiliary sequence $V^n$ is specified from node X to node Y to be jointly typical with $(X^n,Y^n,U^n)$.  This scenario coincides exactly with Lemma \ref{lemma generic example}, and a sufficient rate is $R_V > I(X;V|U,Y)$.  Finally, an action sequence $Z^n$ is specified from node Y to node Z to be jointly typical with $(Y^n,V^n,U^n)$, where $U^n$ is side information known to the encoder and decoder.  We achieve this using a rate $R_Z > I(Y,V;Z|U)$.  Again, because of the symmetry of our encoding scheme, the Strong Markov Lemma (Theorem \ref{theorem strong markov}) tells us that $(X^n,Y^n,U^n,V^n,Z^n)$ will be jointly typical, and therefore, $(X^n,Y^n,Z^n)$ will be jointly typical.

The rates used by this scheme are
\begin{eqnarray*}
R_1 & = & R_{U,1} + R_V > I(X;U,V|Y), \\
R_2 & = & R_{U,2} + R_Z > I(X;U) + I(Y,V;Z|U).
\end{eqnarray*}

\subsection{Empirical Coordination - Converse (Sections \ref{section complete results}, \ref{section partial results})}

In proving outer bounds for the coordination capacity of various networks, a common {\em time mixing} trick is to make use of a random time variable $Q$ and then consider the value of a random sequence $X^n$ at the random time $Q$ using notation $X_Q$.  We first make this statement precise and discuss the implications of such a construction.

Considering a coordination code for a block length $n$.  We assign $Q$ to have a uniform distribution over the set $\{1,...,n\}$, independent of the action sequences in the network.  The variable $X_Q$ is simply a function of the sequence $X^n$ and the variable $Q$; namely, the variable $X_Q$ takes on the value of the $Q$th element in the sequence $X^n$.  Even though all sequences of actions and auxiliary variables in the network are independent of $Q$, the variable $X_Q$ need not be independent of $Q$.

Here we list a couple of key properties of time mixing.

{\em Property 1:  If all elements of a sequence $X^n$ are identically distributed, then $X_Q$ is independent of $Q$.  Furthermore, $X_Q$ has the same distribution as $X_1$.}  Verifying this property is easy when one considers the conditional distribution of $X_Q$ given $Q$.

{\em Property 2:  For a collection of random sequences $X^n$, $Y^n$, and $Z^n$, the expected joint type ${\bf E} P_{X^n,Y^n,Z^n}$ is equal to the joint distribution of the time-mixed variables $(X_Q,Y_Q,Z_Q)$.}

\begin{eqnarray*}
& {\bf E} & P_{X^n,Y^n,Z^n}(x,y,z) \\
& = & \sum_{x^n,y^n,z^n} p(x^n,y^n,z^n) P_{X^n,Y^n,Z^n}(x,y,z) \\
& = & \sum_{x^n,y^n,z^n} p(x^n,y^n,z^n) \frac{1}{n} \sum_{q=1}^n {\bf 1} ((x_q,y_q,z_q)=(x,y,z)) \\
& = & \frac{1}{n} \sum_{q=1}^n \sum_{x^n,y^n,z^n} p(x^n,y^n,z^n) {\bf 1} ((x_q,y_q,z_q)=(x,y,z)) \\
& = & \frac{1}{n} \sum_{q=1}^n p_{X_q,Y_q,Z_q}(x,y,z) \\
& = & \sum_{q=1}^n p_{X_Q,Y_Q,Z_Q|Q}(x,y,z|q)p(q) \\
& = & p_{X_Q,Y_Q,Z_Q}(x,y,z).
\end{eqnarray*}

\subsubsection{Two nodes - Theorem \ref{theorem two nodes}}

Assume that a rate-coordination pair $(R,p(y|x))$ is in the interior of the coordination capacity region ${\cal C}_{p_0}$ for the two-node network of Figure \ref{figure two nodes} with source distribution $p_0(x)$.  For a sequence of $(2^{nR},n)$ coordination codes that achieves $(R,p(y|x))$, consider the induced distribution on the action sequences.

Recall that $I$ is the message from node X to node Y.

\begin{eqnarray*}
n R & \geq & H(I) \\
& \geq & I(X^n;Y^n) \\
& = & \sum_{q=1}^n I(X_q;Y^n|X^{q-1}) \\
& = & \sum_{q=1}^n I(X_q;Y^n,X^{q-1}) \\
& \geq & \sum_{q=1}^n I(X_q;Y_q) \\
& = & n I(X_Q;Y_Q|Q) \\
& \stackrel{a}{=} & n I(X_Q;Y_Q,Q) \\
& \geq & n I(X_Q;Y_Q).
\end{eqnarray*}
Equality $a$ comes from Property 1 of time mixing.

We would like to be able to say that the joint distribution of $X_Q$ and $Y_Q$ is arbitrarily close to $p_0(x)p(y|x)$ for some $n$.  That way we could conclude, by continuity of the entropy function, that $R \geq I(X;Y)$.

The definition of achievability (Definition \ref{definition achievability}) states that
\begin{eqnarray*}
\left\| P_{X^n,Y^n,Z^n}(x,y,z)-p_0(x)p(y,z|x) \right\|_{TV} \longrightarrow 0 \mbox{ in probability}.
\end{eqnarray*}
Because total variation is bounded, this implies that
\begin{eqnarray*}
{\bf E} \left\| P_{X^n,Y^n,Z^n}(x,y,z)-p_0(x)p(y,z|x) \right\|_{TV} \longrightarrow 0.
\end{eqnarray*}
Furthermore, by the Jensen Inequality,
\begin{eqnarray*}
{\bf E} P_{X^n,Y^n,Z^n}(x,y,z) \longrightarrow p_0(x)p(y,z|x).
\end{eqnarray*}
Now Property 2 of time mixing allows us to conclude the argument for Theorem \ref{theorem two nodes}.

\subsubsection{Isolated node - Theorem \ref{theorem isolated}}

No proof is necessary, as this is a special case of the cascade network with $R_2 = 0$.

\subsubsection{Cascade - Theorem \ref{theorem cascade}}

For the cascade network of Figure \ref{figure cascade}, apply the bound from the two-node network twice---once to show that the rate $R_1 \geq I(X;Y,Z)$ is needed even if node Y and node Z are allowed to fully cooperate, and once to show that the rate $R_2 \geq I(X;Z)$ is needed even if node X and node Y are allowed to fully cooperate.

\subsubsection{Degraded source - Theorem \ref{theorem degraded source}}

Assume that a rate-coordination quadruple $(R_1,R_2,R_3,p(z|x,y))$ is in the interior of the coordination capacity region ${\cal C}_{p_0}$ for the degraded source network of Figure \ref{figure degraded source} with source distribution $p_0(x)$ and the degraded relationship $Y_i = f_0(x_i)$.  For a sequence of $(2^{nR_1},2^{nR_2},2^{nR_3},n)$ coordination codes that achieves $(R_1,R_2,R_3,p(z|x,y))$, consider the induced distribution on the action sequences.

Recall that the message from node X to node Y at rate $R_1$ is labeled $I$, the message from node X to node Z at rate $R_2$ is labeled $J$, and the message from node Y to node Z at rate $R_3$ is labeled $K$.  We identify the auxiliary random variable $U$ as the collection of random variables $(K,X^{Q-1},Q)$.

\begin{eqnarray*}
n R_1 & \geq & H(I) \\
& \geq & H(I|Y^n) \\
& \stackrel{a}{=} & H(I,K|Y^n) \\
& \geq & H(K|Y^n) \\
& = & I(X^n;K|Y^n) \\
& = & \sum_{q=1}^n I(X_q;K|Y^n,X^{q-1}) \\
& = & \sum_{q=1}^n I(X_q;K,X^{q-1},Y^{q-1},Y_{q+1}^n|Y_q) \\
& \geq & \sum_{q=1}^n I(X_q;K,X^{q-1}|Y_q) \\
& = & n I(X_Q;K,X^{Q-1}|Y_Q,Q) \\
& \stackrel{b}{=} & n I(X_Q;K,X^{Q-1},Q|Y_Q) \\
& = & n I(X_Q;U|Y_Q).
\end{eqnarray*}
Equality $a$ is justified because the message $K$ is a function of the message $I$ and the sequence $Y^n$.  Equality $b$ comes from Property 1 of time mixing.

\begin{eqnarray*}
n R_2 & \geq & H(J) \\
& \geq & H(J|K) \\
& \stackrel{a}{=} & H(J,Z^n|K) \\
& \geq & H(Z^n|K) \\
& = & I(X^n;Z^n|K) \\
& \geq & \sum_{q=1}^n I(X_q;Z^n|K,X^{q-1}) \\
& \geq & \sum_{q=1}^n I(X_q;Z_q|K,X^{q-1}) \\
& = & n I(X_Q;Z_Q|K,X^{Q-1},Q) \\
& = & n I(X_Q;Z_Q|U).
\end{eqnarray*}
Equality $a$ is justified because the action sequence $Z^n$ is a function of the messages $J$ and $K$.  Equality $b$ comes from Property 1 of time mixing.

\begin{eqnarray*}
n R_3 & \geq & H(K) \\
& = & I(X^n;K) \\
& = & \sum_{q=1}^n I(X_q;K|X^{q-1}) \\
& = & n I(X_Q;K|X^{Q-1},Q) \\
& \stackrel{a}{=} & n I(X_Q;K,X^{Q-1},Q) \\
& = & n I(X_Q;U).
\end{eqnarray*}
Equality $a$ comes from Property 1 of time mixing.

As seen in the proof for the two-node network, the joint distribution of $X_Q$, $Y_Q$, and $Z_Q$ is arbitrarily close to $p_0(x){\bf 1}(y = f_0(x))p(z|x,y)$.  Therefore, since ${\cal C}_{p_0}$ is a closed set, $(R_1,R_2,R_3,p(z|x,y))$ is in the coordination capacity region stated in Theorem \ref{theorem degraded source}.

It remains to bound the cardinality of $U$.  We can use the standard method rooted in the support lemma of \cite{Csiszar81}.  The variable $U$ should have $|{\cal X}||{\cal Z}| - 1$ elements to preserve the joint distribution $p(x,z)$, which in turn preserves $p(x,y,z)$, $H(X)$, and $H(X|Y)$, and three more elements to preserve $H(X|U)$, $H(X|Y,U)$, and $H(X|Z,U)$.

\subsubsection{Broadcast - Theorem \ref{theorem broadcast}}

For the broadcast network of Figure \ref{figure broadcast}, apply the bound from the two-node network three times---once to show that the rate $R_1 \geq I(X;Y)$ is needed and once to show that the rate $R_2 \geq I(X;Z)$ is needed, and finally a third time to show that the sum-rate $R_1 + R_2 = I(X;Y,Z)$ is needed even if node Y and node Z are allowed to fully cooperate.

\subsubsection{Cascade multiterminal - Theorem \ref{theorem relay}}

Assume that a rate-coordination triple $(R_1,R_2,p(z|x,y))$ is in the interior of the coordination capacity region ${\cal C}_{p_0}$ for the cascade multiterminal network of Figure \ref{figure relay} with source distribution $p_0(x,y)$.  For a sequence of $(2^{nR_1},2^{nR_2},n)$ coordination codes that achieves $(R_1,R_2,p(z|x,y))$, consider the induced distribution on the action sequences.

Recall that the message from node X to node Y at rate $R_1$ is labeled $I$, and the message from node Y to node Z at rate $R_2$ is labeled $J$.  We identify the auxiliary random variable $U$ as the collection of random variables $(I,X^{Q-1},Y^{Q-1},Y_{Q+1}^n,Q)$.  This is the same choice of auxiliary variable used by Wyner and Ziv \cite{WynerZiv76}.  Notice that $U$ satisfies the Markov chain properties $U - X_Q - Y_Q$ and $X_Q - (Y_Q,U) - Z_Q$

\begin{eqnarray*}
n R_1 & \geq & H(I) \\
& \geq & H(I|Y^n) \\
& = & I(X^n;I|Y^n) \\
& = & \sum_{q=1}^n I(X_q;I|Y^n,X^{q-1}) \\
& = & \sum_{q=1}^n I(X_q;I,X^{q-1},Y^{q-1},Y_{q+1}^n|Y_q) \\
& = & n I(X_Q;I,X^{Q-1},Y^{Q-1},Y_{Q+1}^n|Y_Q,Q) \\
& \stackrel{a}{=} & n I(X_Q;I,X^{Q-1},Y^{Q-1},Y_{Q+1}^n,Q|Y_Q) \\
& = & n I(X_Q;U|Y_Q).
\end{eqnarray*}
Equality $a$ comes from Property 1 of time mixing.

\begin{eqnarray*}
n R_2 & \geq & H(J) \\
& \geq & I(X^n,Y^n;Z^n) \\
& = & \sum_{q=1}^n I(X_q,Y_q;Z_q|X^{q-1},Y^{q-1}) \\
& = & \sum_{q=1}^n I(X_q,Y_q;Z_q,X^{q-1},Y^{q-1}) \\
& \geq & \sum_{q=1}^n I(X_q,Y_q;Z_q) \\
& = & n I(X_Q,Y_Q;Z_Q|Q) \\
& \stackrel{a}{=} & n I(X_Q,Y_Q;Z_Q,Q) \\
& \geq & I(X_Q,Y_Q;Z_Q).
\end{eqnarray*}
Equality $a$ comes from Property 1 of time mixing.

As seen in the proof for the two-node network, the joint distribution of $X_Q$, $Y_Q$, and $Z_Q$ is arbitrarily close to $p_0(x,y)p(z|x,y)$.  Therefore, since ${\cal C}_{p_0}$ is a closed set, $(R_1,R_2,p(z|x,y))$ is in the coordination capacity region stated in Theorem \ref{theorem relay}.

It remains to bound the cardinality of $U$.  We can again use the standard method of \cite{Csiszar81}.  Notice that $p(x,y,z|u) = p(x|u)p(y|x)p(z|y,u)$ captures all of the Markovity constraints of the outer bound.  Therefore, convex mixtures of distributions of this form are valid for achieving points in the outer bound.  The variable $U$ should have $|{\cal X}||{\cal Y}||{\cal Z}| - 1$ elements to preserve the joint distribution $p(x,y,z)$, which in turn preserves $I(X,Y;Z)$ and $H(X|Y)$, and one more element to preserve $H(X|Y,U)$.

\subsection{Strong Coordination (Section \ref{section strong coordination})}

\subsubsection{No communication - Theorem \ref{theorem no communication strong}}

The network of Figure \ref{figure no communication strong} with no communication generalizes Wyner's common information work \cite{wyner} to three nodes.  Here we provide a sketch of the proof.

The following phenomenon was noticed both by Wyner \cite{wyner} and by Han and Verd\'{u} \cite{han}.  Consider a memoryless channel $p(x|u)$.  A channel input with distribution $p(u)$ induces an output with distribution $p(x) = \sum_u p(u)p(x|u)$.  If the inputs are i.i.d. then the outputs are i.i.d. as well.  Now suppose that instead a channel input sequence $U^n$ is chosen uniformly at random from a set ${\cal M}$ of $2^{nR}$ deterministic sequences.  If $R > I(X;U)$ then the set ${\cal M}$ can be chosen so that the output distribution is arbitrarily close in total variation to the i.i.d. distribution $\prod_{i=1}^n p(x_i)$ for large enough $n$.

Figure \ref{figure no communication strong achievability} illustrates how to achieve the strong coordination capacity region $\underline{\cal C}$ of Theorem \ref{theorem no communication strong}.  Let each decoder simulate a memoryless channel from $U$ to $X$, $Y$, or $Z$, depending on the particular node.  The common randomness $\omega$ is used to index a sequence $U^n(\omega)$ that is used as the inputs to the channels.  Notice that the action sequences $X^n$, $Y^n$, and $Z^n$ produced via these three separate channels are distributed the same as if they were generated as outputs of a single channel because $p(x,y,z|u) = p(x|u)p(y|u)p(z|u)$ according to the definition of $\underline{\cal C}$ in the theorem.  Since $R > I(X,Y,Z;U)$ for points in the interior of $\underline{\cal C}$, this scheme will achieve strong coordination.

\begin{figure}[h]
\psfrag{l1}[][][.9]{$\omega$}
\psfrag{l2}[][][.9]{\;\;\;\;\;\; Common randomness}
\psfrag{l3}[][][.8]{Node X}
\psfrag{l4}[][][.8]{Node Y}
\psfrag{l5}[][][.8]{Node Z}
\psfrag{l6}[][][.8]{$p(x|u)$}
\psfrag{l7}[][][.8]{$p(y|u)$}
\psfrag{l8}[][][.8]{$p(z|u)$}
\psfrag{l9}[][][.8]{$U^n(\omega)$}
\psfrag{l10}[][][.8]{$U^n(\omega)$}
\psfrag{l11}[][][.8]{$U^n(\omega)$}
\psfrag{l12}[][][.8]{$X^n$}
\psfrag{l13}[][][.8]{$Y^n$}
\psfrag{l14}[][][.8]{$Z^n$}
\centerline{\includegraphics[width=.4\textwidth]{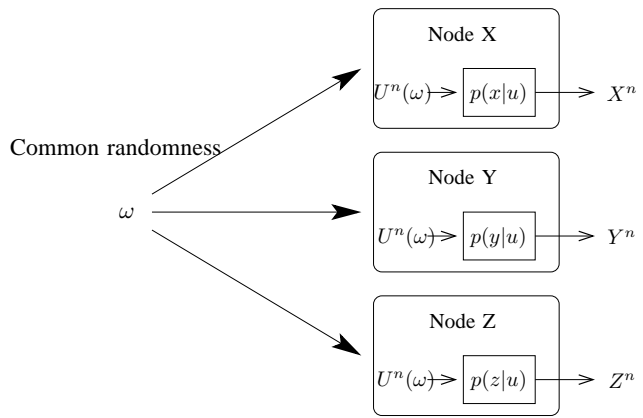}}
\caption{{\em Achievability for no-communication network.}  The strong coordination capacity region $\protect\underline{\cal C}$ of Theorem \ref{theorem no communication strong} is achieved in a network with no communication by using the common randomness to specify a sequence $U^n(\omega)$ that is then passed through a memoryless channel at each node using private randomness.}
\label{figure no communication strong achievability}
\end{figure}

For the converse, identify the auxiliary variable $U$ as $\omega$ and notice that $X_q$, $Y_q$, and $Z_q$ are conditionally independent (for all $q$) given $\omega$.
\begin{eqnarray*}
n R & \geq & H(\omega) \\
& \geq & I(X^n,Y^n,Z^n;\omega) \\
& \geq & I(X^n,Y^n,Z^n;U).
\end{eqnarray*}
Since $X^n$, $Y^n$, and $Z^n$ have a joint distribution close in total variation to the i.i.d. distribution $\prod_{i=1}^n p(x_i,y_i,z_i)$, it can be shown that they can essentially be treated as i.i.d. sequences in the mutual information bounds (see \cite{cuff08}).  If they were i.i.d. we would have
\begin{eqnarray*}
& & I(X^n,Y^n,Z^n;U) \\
& = & \sum_{q=1}^n I(X_q,Y_q,Z_q;U|X^{q-1},Y^{q-1},Z^{q-1}) \\
& = & \sum_{q=1}^n I(X_q,Y_q,Z_q;U,X^{q-1},Y^{q-1},Z^{q-1}) \\
& \geq & \sum_{q=1}^n I(X_q,Y_q,Z_q;U) \\
& \geq & n \min_{\tilde{U}} I(X,Y,Z;\tilde{U}),
\end{eqnarray*}
where the minimization is over all eligible auxiliary $\tilde{U}$ that separate $X$, $Y$, and $Z$ into conditional independence.

It remains to bound the cardinality of $U$.  We can again use the standard method of \cite{Csiszar81}.  The variable $U$ should have $|{\cal X}||{\cal Y}||{\cal Z}| - 1$ elements to preserve the joint distribution $p(x,y,z)$, which in turn preserves $H(X,Y,Z)$, and one more element to preserve $H(X,Y,Z|U)$.

\subsubsection{Two nodes - Theorem \ref{theorem two nodes strong}}

The strong coordination capacity region for the two-node network of Figure \ref{figure two nodes strong} is the main result of \cite{cuff08}:
\begin{eqnarray}
\label{equation two nodes strong capacity}
\underline{\cal C}_{p_0} & = & \left\{
\begin{array}{l}
p(y|x) \; : \; \exists p(u|x,y) \mbox{ such that} \\
p(x,y,u) = p(u)p(x|u)p(y|u) \\
|{\cal U}| \leq |{\cal X}||{\cal Y}| + 1, \\
R \geq I(X;U), \\
R_0 + R \geq I(X,Y;U).
\end{array}
\right\},
\end{eqnarray}
where $R_0$ refers to the rate of common randomness, and $R$ refers to the communication rate.

In the case of no common randomness ($R_0 = 0$), the stronger inequality in (\ref{equation two nodes strong capacity}) on the rate $R$ become the second, $R \geq I(X,Y;U)$.  Because of the Markov constraint on $U$, the minimum value of the right-hand side of this inequality is Wyner's common information $C(X;Y)$.

Additionally, Theorem \ref{theorem two nodes strong} states that if $R_0$ is greater than the necessary conditional entropy $H(Y \dag X)$ then rates $R > I(X;Y)$ are sufficient for achieving strong coordination.  This is a straightforward application of the definition of $H(Y \dag X)$.  We can verify this with the following choice of $U$:
\begin{eqnarray*}
U & = & \argmin_{f(Y) \; : \; X - f(Y) - Y} H(f(Y)|X).
\end{eqnarray*}
Notice that this choice of $U$ separates $X$ and $Y$ into a Markov chain by definition.  Also, the mutual information $I(X;U)$ is less than or equal to $I(X;Y)$, since $U$ is a function of $Y$, thus satisfying the first rate inequality in (\ref{equation two nodes strong capacity}).  The second inequality is satisfied because of the chain rule,
\begin{eqnarray*}
I(X,Y;U) & = & I(X;U) + I(Y;U|X) \\
& = & I(X;U) + H(Y \dag X) \\
& \leq & I(X;Y) + H(Y \dag X).
\end{eqnarray*}

Furthermore, we can show that this is the least amount of common randomness needed to fully expand the strong coordination capacity region.  In other words, the minimum $R_0$ such that $(R_0,I(X;Y))$ is in the strong rate-coordination region $\underline{\cal R}_{p_0}p(y|x)$ is $H(Y \dag X)$.

To prove this, first consider the implications of $R = I(X;Y)$.  This means that in order to satisfy the first rate inequality in (\ref{equation two nodes strong capacity}), we must have $I(X;U) \leq I(X;Y)$.  However, because of the Markovity, $I(X;U) = I(X;U,Y)$.  Therefore, $I(X;U|Y) = 0$, which implies a second Markov condition $X - Y - U$ in addition to $X - U - Y$.

We are concerned with minimizing the required rate of common randomness $R_0$.  Since $R = I(X;Y)$, the second rate inequality in (\ref{equation two nodes strong capacity}) becomes $R_0 \geq I(Y;U|X)$.  The conditional entropy $H(Y|X)$ is fixed, so we want to maximize the conditional entropy $H(Y|U,X)$.

With the distribution $p(x|y)$ in mind, we can clump values of $Y$ together for which the channel from $Y$ to $X$ is identical.  Define a function $f$ with the property that
\begin{eqnarray}
\label{equation same conditional}
f(y) = f(\tilde{y}) & \Longleftrightarrow & p(x|y) = p(x|\tilde{y}) \mbox{ for } \forall x \in {\cal X}.
\end{eqnarray}
Letting $U = f(Y)$ will be the choice of $U$ that simultaneously maximizes $H(Y|U,X)$ and satisfies the Markov conditions $X - U - Y$ and $X - Y - U$.  We can compare $U$ to any other choice $\tilde{U}$ that satisfies the conditions and show that the resulting conditional entropy $H(Y|\tilde{U},X)$ is smaller.

Another way to state the two Markov conditions is that for all values of $y$ and $\tilde{u}$ such that $p(y,\tilde{u}) > 0$, the conditional distributions $p(x|y)$ and $p(x|\tilde{u})$ are equal because $p(x|y) = p(x|y,\tilde{u}) = p(x|\tilde{u})$.  Notice that the value of $U = f(Y)$, characterized in (\ref{equation same conditional}), only depends on the channel $p(x|y)$.  However, with probability one the value of $U$ can be determined from $\tilde{U}$ based on the conditional distribution $p(x|\tilde{u})$.  Therefore,
\begin{eqnarray*}
H(Y|\tilde{U},X) & = & H(Y,U|\tilde{U},X) \\
& = & H(Y|U,\tilde{U},X) + I(U|\tilde{U},X) \\
& = & H(Y|U,\tilde{U},X) \\
& \leq & H(Y|U,X).
\end{eqnarray*}

\subsection{Rate-distortion theory (Sections \ref{section rate distortion})}

We establish the relationship from Theorem \ref{theorem rate distortion} between the coordination capacity region and the rate-distortion region in two parts.  First we show that ${\cal D}_{p_0}$ contains $A {\cal C}_{p_0}$ and then the other way around.  To keep clutter to a minimum and without loss of generality, we only discuss a single distortion measure $d$, rate $R$, and a pair of sequences of actions $X^n$ and $Y^n$.

\subsubsection{Coordination implies distortion (${\cal D}_{p_0} \supset A {\cal C}_{p_0}$)}

The distortion incurred with respect to a distortion function $d$ on a set of sequences of actions is a function of the joint type of the sequences.  That is,
\begin{eqnarray}
\label{equation distortion type}
d^{(n)}(x^n,y^n) & = & \frac{1}{n} \sum_{i=1}^n d(x_i,y_i) \nonumber \\
& = & \frac{1}{n} \sum_{i=1}^n \sum_{x,y} {\bf 1}(x_i = x,y_i = y) d(x,y) \nonumber \\
& = & \sum_{x,y} d(x,y) \frac{1}{n} \sum_{i=1}^n {\bf 1}(x_i = x,y_i = y) \nonumber \\
& = & \sum_{x,y} d(x,y) P_{x^n,y^n}(x,y) \nonumber \\
& = & {\bf E}_{P_{x^n,y^n}} d(X,Y).
\end{eqnarray}

When a rate-coordination tuple $(R,p(x,y))$ is in the interior of the coordination capacity region ${\cal C}_{p_0}$, we are assured the existence of a coordination code for any $\epsilon > 0$ for which
\begin{eqnarray*}
\Pr (\| P_{X^n,Y^n} - p \|_{TV} > \epsilon) & < & \epsilon.
\end{eqnarray*}
Therefore, with probability greater that $1 - \epsilon$,
\begin{eqnarray*}
{\bf E}_{P_{X^n,Y^n}} d(X,Y) & \leq & {\bf E}_{p} d(X,Y) + \epsilon d_{max}.
\end{eqnarray*}
Recalling (\ref{equation distortion type}) yields,
\begin{eqnarray*}
{\bf E} d^{(n)}(x^n,y^n) & \leq & {\bf E}_{p} d(X,Y) + 2\epsilon d_{max}.
\end{eqnarray*}
As expected, a sequence of $(2^{nR},n)$ coordination codes that achieves empirical coordination for the joint distribution $p(x,y)$ also achieves the point in the rate-distortion region with the same rate and with distortion value ${\bf E}_p d(X,Y)$.

\subsubsection{Distortion implies coordination (${\cal D}_{p_0} \subset A {\cal C}_{p_0}$)}

Suppose that a $(2^{nR},n)$ rate-distortion codes achieves distortion ${\bf E} d^{(n)}(X^n,Y^n) \leq D$.  Substituting from (\ref{equation distortion type}),
\begin{eqnarray*}
{\bf E} \left[ {\bf E}_{P_{X^n,Y^n}} d(X,Y) \right] & \leq & D.
\end{eqnarray*}
However,
\begin{eqnarray*}
{\bf E} \left[ {\bf E}_{P_{X^n,Y^n}} d(X,Y) \right] & = & {\bf E}_{{\bf E} P_{X^n,Y^n}} d(X,Y)
\end{eqnarray*}
by linearity.

We can achieve the rate-coordination pair $(R, {\bf E} P_{X^n,Y^n})$ by augmenting the rate-distortion code.  If we repeat the use of the rate-distortion code over $k$ blocks of length $n$ each, then we induce a joint distribution on $(X^{kn},Y^{kn})$ that consists of i.i.d. sub-blocks $(X^n,Y^n),...,(X_{kn-n+1}^{kn},Y_{kn-n+1}^{kn})$ denoted as $(X^{(1)n},Y^{(1)n}),...,(X^{(k)n},Y^{(k)n})$.

By the weak law of large number,
\begin{eqnarray*}
P_{X^{kn},Y^{kn}} & = & \frac{1}{k} \sum_{i=1}^k P_{X^{(i)n},Y^{(i)n}} \\
& \longrightarrow & {\bf E} P_{X^n,Y^n} \mbox{ in probability}.
\end{eqnarray*}
Point-wise convergence in probability implies that as $k$ grows
\begin{eqnarray*}
\left\| P_{X^{kn},Y^{kn}} - {\bf E} P_{X^n,Y^n} \right\|_{TV} & \longrightarrow & 0 \mbox{ in probability}.
\end{eqnarray*}

Thus, for any point $(R,D)$ in the rate-distortion region we have identified an associated point $(R, {\bf E} P_{X^n,Y^n})$ in the coordination-capacity region.  Indeed, the rate-distortion region is a linear projection of the coordination-capacity region.

%% file: coordination_conclusion.tex
\section{Remarks}
\label{section conclusion}

Rather than inquire about the possibility of moving data in a network, we have asked for the set of all achievable joint distribution on actions at the nodes.  For some three-node networks we have fully characterized the answer to this question, while for others we have established bounds.

Some of the results discussed in this work extend nicely to larger networks.  Consider for example an extended cascade network shown in Figure \ref{figure cascade large}, where $X$ is given randomly by nature and $Y_1$ through $Y_{k-1}$ are actions based on a cascade of communication.  Just as in the cascade network of Section \ref{subsection cascade}, we can achieve rates $R_i \geq I(X;Y_i,...,Y_k)$ for empirical coordination by sending messages to the last nodes in the chain first and conditioning later messages on earlier ones.  These rates meet the cut-set bound.  We now can make an interesting observation about assigning unique tasks to nodes in such a network.  Suppose $k$ tasks are to be completed by the $k$ nodes in this cascade network, one at each node.  Node X is assigned a task randomly, and the communication in the network is used to assign a permutation of all the tasks to the nodes in the network.  The necessary rates in the network are $R_i \geq \log(\frac{k}{i})$.  The sum of all the rates in the network, for large $k$, is then approximately $R_{total} \geq k$ nats, where $k$ is the number of tasks and nodes in the network.

\begin{figure}[h]
\psfrag{l1}[][][.8]{$R_1$}
\psfrag{l2}[][][.8]{$R_2$}
\psfrag{l3}[][][.8]{$R_{k-1}$}
\psfrag{l4}[][][.8]{$X \sim p_0(x)$}
\psfrag{l5}[][][.8]{$Y_1$}
\psfrag{l6}[][][.8]{$Y_2$}
\psfrag{l7}[][][.8]{$Y_{k-2}$}
\psfrag{l8}[][][.8]{$Y_{k-1}$}
\centerline{\includegraphics[width=.5\textwidth]{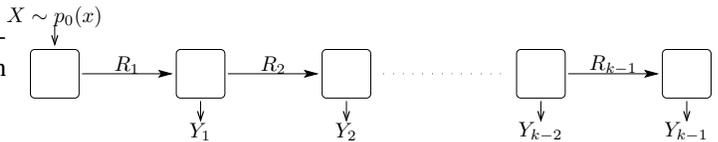}}
\caption{{\em Extended cascade network.}  This is an extension of the cascade network of Section \ref{subsection cascade}.  Action $X$ is given randomly by nature according to $p_0(x)$, and a cascade of communication is used to produce actions $Y_1$ through $Y_{k-1}$.  The coordination capacity region contains all rate-coordination tuples that satisfy $R_i \geq I(X;Y_i,...,Y_k)$ for all $i$.  In particular, the sum rate needed to assign a permutation of $k$ tasks to the $k$ nodes grows linearly with the number of nodes.}
\label{figure cascade large}
\end{figure}

Now consider the same task assignment scenario for an extended broadcast network shown in Figure \ref{figure broadcast large}.  Here again $X$ is given randomly by nature, but $Y_1$ through $Y_{k-1}$ are actions based on individual messages sent to each of the nodes.  Again, we want to assign a permutation of all the $k$ tasks to all of the $k$ nodes.  We can use ideas from the broadcast network results of Section \ref{subsection broadcast}.  For example, let us assign default tasks to the nodes so that $Y_1 = 1,...,Y_{k-1} = k-1$ unless told otherwise.  Now the communication is simply used to tell each node when it must choose task $k$ rather than the default task, which will happen about one time out of $k$.  The rates needed for this scheme are $R_i \geq H(1/k)$, where $H$ is the binary entropy function.  For large $k$, the sum of all the rates in the network is approximately $R_{total} \geq \ln k + 1$ nats.  The cut-set bound gives us a lower bound on the sum rate of $R_{total} \geq \ln k$ nats.  Therefore, we can conclude that the optimal sum rate scales with the logarithm of the number of nodes in the network.

\begin{figure}[h]
\psfrag{l1}[][][.8]{$\;\;\;\;\;\;\; X \sim p_0(x)$}
\psfrag{l2}[][][.8]{$Y_1$}
\psfrag{l3}[][][.8]{$Y_2$}
\psfrag{l4}[][][.8]{$Y_3$}
\psfrag{l5}[][][.8]{$Y_4$}
\psfrag{l6}[][][.8]{$Y_{k-2}$}
\psfrag{l7}[][][.8]{$Y_{k-1}$}
\psfrag{l8}[][][.8]{}
\psfrag{l9}[][][.8]{}
\psfrag{l10}[][][.8]{}
\psfrag{l11}[][][.8]{}
\psfrag{l12}[][][.8]{}
\psfrag{l13}[][][.8]{}
\centerline{\includegraphics[width=.3\textwidth]{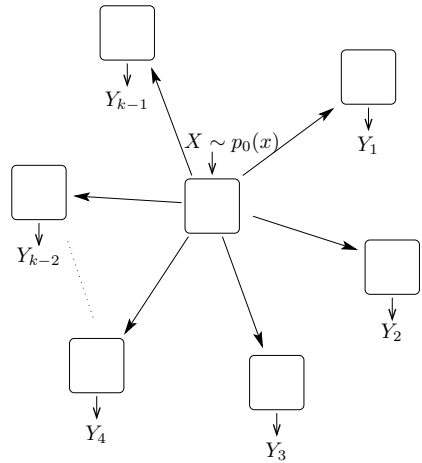}}
\caption{{\em Extended broadcast network.}  This is an extension of the broadcast network of Section \ref{subsection broadcast}.  Action $X$ is given randomly by nature according to $p_0(x)$, and each peripheral node produces an action $Y_i$ based on an individual message at rate $R_i$.  Bounds on the coordination capacity region show that the sum rate needed to assign a permutation of $k$ tasks to the $k$ nodes grows logarithmically with the number of nodes.}
\label{figure broadcast large}
\end{figure}

Even without explicitly knowing the coordination capacity region for the broadcast network, we are able to use bounds to establish the scaling laws for the total rate needed to assign tasks uniquely, and we can compare the efficiency of the broadcast network (logarithmic in the network size) with that of the cascade network (linear in the network size) for this kind of coordination.

We would also like to understand the coordination capacity region for a noisy network.  For example, the communication capacity region for the broadcast channel $p(\tilde{y}_1, \tilde{y}_2 | \tilde{x})$ of Figure \ref{figure broadcast channel} has undergone serious investigation.  The standard question is, how many bits of independent information can be communicated from $X$ to $Y_1$ and from $X$ to $Y_2$.  We know the answer if the broadcast channel is degraded; that is, if $Y_2$ can be viewed as a noisy version of $Y_1$.  We also know the answer if the channel can be separated into two orthogonal channels or is deterministic.  But what if instead we are trying to coordinate actions via the broadcast channel, similar to the broadcast network of Section \ref{subsection broadcast}?  Now we care about the dependence between $Y_1$ and $Y_2$.  The broadcast channel will impose a natural dependence between the channel outputs $\tilde{Y}_1$ and $\tilde{Y}_2$ that we abolish if we try to send independent information to the two nodes.  After all, the communication capacity region for the broadcast channel depends only on the marginals  $p(\tilde{y}_1 | \tilde{x})$ and $p(\tilde{y}_2 | \tilde{x})$.  Here we are wasting a valuable resource---the natural conditional dependence between $\tilde{Y}_1$ and $\tilde{Y}_2$ given $\tilde{X}$.

\begin{figure}[h]
\psfrag{X}[][][.8]{$X$}
\psfrag{Y2}[][][.8]{$Y_2$}
\psfrag{Y1}[][][.8]{$Y_1$}
\psfrag{A}[][][.8]{$\tilde X$}
\psfrag{C}[][][.8]{$\tilde Y_2$}
\psfrag{B}[][][.8]{$\tilde Y_1$}
\psfrag{Encoder}[][][.8]{Encoder}
\psfrag{Decoder1}[][][.8]{Decoder}
\psfrag{Decoder2}[][][.8]{Decoder}
\psfrag{Channel}[][][.8]{Channel}
\psfrag{Dist3}[][][.8]{$p(\tilde y_1,\tilde y_2|\tilde x)$}
\psfrag{Dist}[][][.8]{$p_0(x)$}
\psfrag{Dist2}[][][.8]{}
\centerline{\includegraphics[width=.5\textwidth]{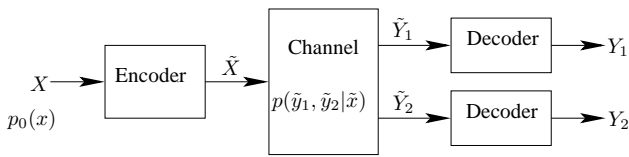}}
\caption{{\em Broadcast channel.}  When a noisy channel is used to coordinate joint actions $(X,Y_1,Y_2)$, what is the resulting coordination capacity region?  The broadcast network of Section \ref{subsection broadcast} is a noiseless special case.}
\label{figure broadcast channel}
\end{figure}

Again, we are enlarging the focus from communication of independent information to the creation of coordinated actions.  This larger question may force a simpler solution and illuminate the problem of independent information (the standard channel capacity formulation) as a special case.  Presumably, information is being communicated for a reason---so future cooperative behavior can be achieved.

\section{Final Remarks}
\label{section final remarks}

At first it seems that the nodes in a network can cooperate arbitrarily without communication.  Prior arrangement achieves that.  Also common randomness achieves it.

But the problem changes dramatically when some of the nodes take actions specified by nature.  Now some communication to the remaining nodes becomes necessary to establish the desired dependence.

We have established the rate-dependence tradeoff for cascade networks and isolated node networks found in Section \ref{section complete results}.  The broadcast network of Figure \ref{figure broadcast} remains elusive, perhaps for the same reason that the broadcast channel is difficult.

%% file: Coordination.bbl
\begin{thebibliography}{10}

\bibitem{network_coding}
R.~Ahlswede, N.~Cai, S.-Y. Li, and R.~Yeung.
\newblock Network information flow.
\newblock {\em IEEE Trans. on Info. Theory}, 46(4):1204--1216, July 2000.

\bibitem{tsitsiklis86}
J.~Tsitsiklis, D.~Bertsekas, and M.~Athans.
\newblock Distributed asynchronous deterministic and stochastic gradient
  optimization algorithms.
\newblock {\em IEEE Trans. on Automatic Control}, 31(9):803--812, Sept. 1986.

\bibitem{xiao-boyd-kim}
L.~Xiao, S.~Boyd, and S.-J. Kim.
\newblock Distributed average consensus with least-mean-square deviation.
\newblock {\em Journal of Parallel and Distributed Computing}, 67(1):33--46,
  Jan. 2007.

\bibitem{bollobas}
B.~Bollobas.
\newblock {\em The Art of Mathematics: Coffee Time in Memphis}.
\newblock Cambridge University Press, 2006.

\bibitem{yao79}
A.~Yao.
\newblock Some complexity questions related to distributive
  computing(preliminary report).
\newblock In {\em ACM Symposium on Theory of Computing}, pages 209--213, 1979.

\bibitem{orlitsky90}
A.~Orlitsky and A.~El~Gamal.
\newblock Average and randomized communication complexity.
\newblock {\em IEEE Trans. on Info. Theory}, 36(1):3--16, Jan. 1990.

\bibitem{ayaso-shah-dahleh08}
O.~Ayaso, D.~Shah, and M.~Dahleh.
\newblock Distributed computation under bit constraints.
\newblock In {\em IEEE Conference on Decision and Control}, pages 4837--4842,
  Dec. 2008.

\bibitem{Cover07Permuter}
T.~Cover and H.~Permuter.
\newblock Capacity of coordinated actions.
\newblock In {\em IEEE International Symp. on Info. Theory}, Nice, 2007.

\bibitem{wyner}
A.~Wyner.
\newblock The common information of two dependent random variables.
\newblock {\em IEEE Trans. on Info. Theory}, 21(2):163--179, March 1975.

\bibitem{AnantharamBorkar2007}
V.~Anantharam and V.~Borkar.
\newblock Common randomness and distributed control; a counterexample.
\newblock {\em Systems and Control Letters}, 56:568--572, 2007.

\bibitem{Barnum_quantum_mixed_states}
H.~{Barnum}, C.~{Caves}, C.~{Fuchs}, R.~{Jozsa}, and B.~{Schumacher}.
\newblock On quantum coding for ensembles of mixed states.
\newblock {\em Journal of Physics A: Mathematical and General}, 34:6767--6785,
  2001.

\bibitem{kramer_savari07}
G.~Kramer and S.~Savari.
\newblock Communicating probability distributions.
\newblock {\em IEEE Trans. on Info. Theory}, 53(2):518--525, Feb. 2007.

\bibitem{weissman-ordentlich}
T.~Weissman and E.~Ordentlich.
\newblock The empirical distribution of rate-constrained source codes.
\newblock {\em IEEE Trans. on Info. Theory}, 51(11):3718--3733, Nov. 2005.

\bibitem{han}
T.~Han and S.~Verd{\'u}.
\newblock Approximation theory of output statistics.
\newblock {\em IEEE Trans. on Info. Theory}, 39(3):752--772, May 1993.

\bibitem{bennett2002}
C.~Bennett, P.~Shor, J.~Smolin, and A.~Thapliyal.
\newblock Entanglement-assisted capacity of a quantum channel and the reverse
  shannon theorem.
\newblock {\em IEEE Trans. on Info. Theory}, 48(10):2637--2655, Oct. 2002.

\bibitem{Shannon60}
C.~Shannon.
\newblock Coding theorems for a discrete source with fidelity criterion.
\newblock In R.~Machol, editor, {\em Information and Decision Processes}, pages
  93--126. 1960.

\bibitem{Yamamoto_source_coding81}
H.~Yamamoto.
\newblock Source coding theory for cascade and branching communication systems.
\newblock {\em IEEE Trans. on Info. Theory}, 27(3):299--308, May 1981.

\bibitem{Kaspi_berger82}
A.~Kaspi and T.~Berger.
\newblock Rate-distortion for correlated sources with partially separated
  encoders.
\newblock {\em IEEE Trans. on Info. Theory}, 28:828--840, Nov. 1982.

\bibitem{barros-servetto04}
J.~Barros and S.~Servetto.
\newblock A note on cooperative multiterminal source coding.
\newblock In {\em Conference on Information Sciences and Systems}, March 2004.

\bibitem{Yamamoto_triangle96}
H.~Yamamoto.
\newblock Source coding theory for a triangular communication system.
\newblock {\em IEEE Trans. on Info. Theory}, 42(3):848--853, May 1996.

\bibitem{WolfWynerZiv}
J.~Wolf, A.~Wyner, and J.~Ziv.
\newblock Source coding for multiple descriptions.
\newblock {\em Bell System Technical Journal}, 59:1417--1426, Oct. 1980.

\bibitem{zhang87Berger}
Z.~Zhang and T.~Berger.
\newblock New results in binary multiple descriptions.
\newblock {\em IEEE Trans. on Info. Theory}, 33:502--521, July 1987.

\bibitem{Berger78}
T.~Berger.
\newblock Multiterminal source coding.
\newblock In G.~Longo, editor, {\em Information Theory Approach to
  Communications}, pages 171--231. CISM Course and Lecture, 1978.

\bibitem{vasudevan}
D.~Vasudevan, C.~Tian, and S.~Diggavi.
\newblock Lossy source coding for a cascade communication system with
  side-informations.
\newblock In {\em Allerton Conference on Communication, Control, and
  Computing}, Sep. 2006.

\bibitem{Gu06}
W.~Gu and M.~Effros.
\newblock On multi-resolution coding and a two-hop network.
\newblock In {\em Data Compression Conference}, 2006.

\bibitem{Bakshi07}
M.~Bakshi, M.~Effros, W.~Gu, and R.~Koetter.
\newblock On network coding of independent and dependent sources in line
  networks.
\newblock In {\em IEEE International Symp. on Info. Theory}, Nice, 2007.

\bibitem{WynerZiv76}
A.~Wyner and J.~Ziv.
\newblock The rate-distortion function for source coding with side information
  at the decoder.
\newblock {\em IEEE Trans. on Info. Theory}, 22(1):1--10, Jan. 1976.

\bibitem{Cuff09}
P.~Cuff, H.~Su, and A.~El Gamal.
\newblock Cascade multiterminal source coding.
\newblock In {\em IEEE International Symp. on Info. Theory}, Seoul, 2009.

\bibitem{Orlitsky01}
A.~Orlitsky and J.~Roche.
\newblock Coding for computing.
\newblock {\em IEEE Trans. on Info. Theory}, 47(3):903--917, March 2001.

\bibitem{cuff08}
P.~Cuff.
\newblock Communication requirements for generating correlated random
  variables.
\newblock In {\em IEEE International Symp. on Info. Theory}, pages 1393--1397,
  Toronto, 2008.

\bibitem{steinberg}
Y.~Steinberg and S.~Verd{\'u}.
\newblock Simulation of random processes and rate-distortion theory.
\newblock {\em IEEE Trans. on Info. Theory}, 42(1):63--86, Jan. 1996.

\bibitem{devetak}
I.~Devetak, A.~Harrow, P.~Shor, A.~Winter, and C.~Bennett.
\newblock Quantum reverse shannon theorem.
\newblock Presentation:
  http://www.research.ibm.com/people/b/bennetc/QRSTonlineVersion.pdf, 2007.

\bibitem{CovThom06}
T.~Cover and J.~Thomas.
\newblock {\em Elements of Information Theory}.
\newblock Wiley, New York, 2nd edition, 2006.

\bibitem{Wyner75_WAK}
A.~Wyner.
\newblock On source coding with side-information at the decoder.
\newblock {\em IEEE Trans. on Info. Theory}, 21(3):294--300, May 1975.

\bibitem{Csiszar81}
I.~Csisz{\'a}r and J.~K{\"o}rner.
\newblock {\em Information Theory: Coding Theorems for Discrete Memoryless
  Systems}.
\newblock Academic, New York, 1981.

\end{thebibliography}
